\documentclass[onefignum,onetabnum]{siamart190516}

\usepackage{lipsum}
\usepackage{amsfonts}
\usepackage{graphicx}
\usepackage{epstopdf}
\usepackage{algorithmic}
\ifpdf
  \DeclareGraphicsExtensions{.eps,.pdf,.png,.jpg}
\else
  \DeclareGraphicsExtensions{.eps}
\fi

\usepackage{enumitem}
\setlist[enumerate]{leftmargin=.5in}
\setlist[itemize]{leftmargin=.5in}

\newsiamremark{remark}{Remark}
\newsiamremark{hypothesis}{Hypothesis}
\crefname{hypothesis}{Hypothesis}{Hypotheses}
\newsiamthm{claim}{Claim}

\headers{Regression analysis of distributional data}{A. Karimi and T. T. Georgiou}

\title{Regression analysis of distributional data\\
through Multi-Marginal Optimal transport\thanks{
\funding{This research was partially supported by 
NSF under grants 1807664, 1839441, and
AFOSR under grant FA9550-20-1-0029.}}}

\author{Amirhossein Karimi\thanks{Department of Mechanical
  and Aerospace Engineering,
  University of California, Irvine,
  California, USA
  (\email{amirhosk@uci.edu}).}
\and Tryphon T. Georgiou\thanks{Department of Mechanical
  and Aerospace Engineering,
  University of California, Irvine,
  California, USA
  (\email{tryphon@uci.edu}).}
}

\usepackage{amsopn}

\usepackage{subfigure}
\usepackage{amsmath,amssymb,euscript,yfonts,psfrag,latexsym,dsfont,graphicx}
\usepackage{bbm,color,amstext,wasysym,parskip,balance,cite}
\usepackage{comment}

\def\spacingset#1{\def\baselinestretch{#1}\small\normalsize}
\spacingset{1}

\graphicspath{{./},{./figures/}}

\newcommand{\PFO}{Perron-Frobenius operator}
\newcommand{\X}{\mathbb{X}} %The state-space
\newcommand{\Rd}{\mathbb{R}^d}
\newcommand{\R}{\mathbb{R}}
\newcommand{\PD}{\mathcal{P}_2(\mathbb{X})}
\newcommand{\dd}{{\rm d}}
\begin{document}

\maketitle

\begin{abstract}
  We formulate and solve a regression problem with time-stamped distributional data.
Distributions are considered as points in the Wasserstein space of probability measures, metrized by the $2$-Wasserstein metric, and may represent images, power spectra, point clouds of particles, and so on.
The regression seeks a curve in the Wasserstein space that passes closest to the dataset.
Our regression problem allows utilizing general curves in a Euclidean setting  (linear, quadratic, sinusoidal, and so on), lifted to corresponding measure-valued curves in the Wasserstein space. It can be cast as a multi-marginal optimal transport problem that allows efficient computation. Illustrative academic examples are presented.
\end{abstract}


% REQUIRED
%\begin{AMS}
%  	62Jxx, 62-07, 62G05
%\end{AMS}

\section{Introduction}
Regression analysis seeks a functional dependence of a set of variables with respect to an independent one, or possibly more, by suitable minimization of residuals. In our case the dependent variables are probability distributions and the independent is time. Thereby, the functional dependence may be seen as an estimate of underlying dynamics and flow, while the distributions may represent populations of indistinguishable particles\cite{nichols2017using,haasler2019estimating}, longitudinal succession of images~\cite{niethammer2011geodesic}, spectral densities~\cite{jiang2011geometric}, traffic data~\cite{hong2014geodesic} and so on. 

In order to quantify distance between distributions, it is natural to consider the Wasserstein metric of optimal mass transport. The reason is that 
in this metric, small spatial displacement of mass results in small deviation in the values of the metric (weak continuity)~\cite{villani2003topics}. Thus, we seek to quantify residuals in regression analysis of distributions using the Wasserstein metric. An added insight, when using the so-called $2$-Wasserstein metric, is that regression of Dirac masses reduces to ordinary regression in a Euclidean setting (see Proposition \ref{Euclid_geod_consistency}).

In recent years, the Wasserstein metric has seen a rapidly increasing range of applications in subjects such as computer vision~\cite{haker2004optimal}, machine learning~\cite{liu2019wasserstein}, data fusion~\cite{agueh2011barycenters,chen2018measure,benamou2019second,haasler2019estimating}, and mathematical physics \cite{jordan1998variational,chen2019stochastic}, to name a few. In particular we want to bring attention to recent attempts to develop a framework for principal component analysis in the Wasserstein space~\cite{bigot2017geodesic,cazelles2018geodesic,wang2013linear,seguy2015principal} as it parallels the present work.
In this context, in spite of great many analogies with the Euclidean structure, a salient feature of the Wasserstein geometry on distributions is the curvature of the space. As a result, for instance, a barycenter of a collection of distributions may fail to lie on the first principal component~\cite{cazelles2018geodesic}, and optimization problems are often non-convex and computational demanding.

In this work we seek to determine, via regression, curves in the Wasserstein space as models for given sequences of time-stamped distributions. 
Related problems were considered in \cite{jiang2011geometric,chen2018measure,benamou2019second}. In particular, \cite{jiang2011geometric} considered approximating time-stamped spectral distributions of a non-stationary time series by a Wasserstein geodesic, as a means to obtain a non-parametric model for the underlying dynamics. This problem turned out to be nonlinear and computationally demanding. References \cite{chen2018measure,benamou2019second} proposed interpolation using splines in Wasserstein space. 
The present work follows a similar path in that we seek measure-valued curves, e.g., linear, quadratic, and so on, and a probability law on such curves. However, instead, we consider regression problems where the marginals approximate the distributional snapshots. Our approach may be seen as generalizing least-squares regression problems in Wasserstein space. We show that using a multi-marginal optimal transport formulation~\cite{pass2015multi}, regression problems can be recast as a linear program. Thereby, Sinkhorn's algorithm can be employed to solve efficiently  entropy-regularized versions.

The structure and contributions in this paper are as follows. Section \ref{sec:preliminaries} provides background on the theory of optimal mass transport that underlies this work. 
In Section \ref{sec:regression}, the concept of measure-valued curves is described and our regression problem is formulated. 
In Section \ref{sec:multimarginal}, a multi-marginal formulation of the regression problem is proposed, followed by Section \ref{sec:discretization}, where discretization and a generalized Sinkhorn algorithm to solve the multi-marginal problem are given. 
Sections \ref{sec:Gaussian} and \ref{sec:GMM} present case studies based on Gaussian distributions and Gaussian mixture models, respectively. Section \ref{sec:invariant} details a potential application of the framework to estimate the invariant measure of a dynamical system by extrapolating the distribution along the flow obtained by the regression on specified snapshots. We conclude in Section \ref{sec:conclusion} with remarks and ongoing work.

%We will also delineate how to use the method of this paper to estimate the Perron-Frobenius operators and invariant measures associated with dynamical systems. The long-term behavior of a dynamical system can be characterized by its associated invariant measure supported on an invariant set. Indeed, this is the fixed point of Perron-Frobenius operator which pushes forward distributions under the action of dynamics~\cite{dellnitz1997exploring}. The invariant sets, for example, can represent equilibrium points, periodic and quasi-periodic orbits sitting on some lower-dimensional manifolds~\cite{junge2017sighting}. In this study, we estimate the Perron-Frobenius operator and its corresponding invariant measure without hypothesizing any information on the dynamics and by relying solely on a few available distributional snapshots. One approach to solve this problem is presented in \cite{karimi2020data} in which a (non)linear dynamics is sough such that its corresponding Perron-Frobenius operator tracks approximately the distributional snapshots. There are a plethora of numerical methods to compute invariant measure and sets, most of which conducted for known dynamics, or where the pointwise correspondence between the successive points in time is available (See \cite{korda2021convex} and references therein). 

\section{Preliminaries on optimal mass transport}\label{sec:preliminaries}

We herein provide background on the theory of optimal mass transport (OMT) that underlies the developments in the body of the paper, and refer to  \cite{villani2003topics,villani2008optimal,ambrosio2013user} for more detailed exposition.

Let $\X= \Rd$ be equipped with the Borel $\sigma$-algebra $\mathcal{B}(\X)$.
Let $\mu_0$ and $\mu_1$ be two probability measures in $\PD$, the space of probability measures with finite second moments. We consider the problem to minimize the quadratic cost
\[
  \int_{\X} \|T(x)-x\|_2^2 ~\dd\mu_0(x)   
\]
over the space of transport maps
\begin{align*}
    T: \quad &\X \rightarrow \X\\ 
    &x \mapsto T(x)
\end{align*}
that are measurable and ``push forward'' $\mu_0$ to $\mu_1$, a property written as $T_\#\mu_0 = \mu_1$. This means that for all $A\in \mathcal{B}(\X)$, we have $\mu_1(A)=\mu_0(T^{-1}(A))$ or, equivalently, that for all integrable functions $f(x)$ with respect to $\mu_1$,
\begin{equation}
\label{push_forward}
\int_\X f(x)\dd \mu_1(x)= \int_\X f(T(x)) \dd \mu_0(x).
\end{equation}
If $\mu_0$ is absolutely continuous with respect to the Lebesgue measure, it is known that the optimal transport problem has a unique solution $\hat{T}(x)$ which turns out to be the gradient of a convex function $\phi(x)$, i.e., $\hat{T}(x)=\nabla\phi(x)$.

The problem is nonlinear and, in general, the optimal transport map may not exist. To this end, in 1942, Kantorovich introduced a relaxed formulation in which, instead of a transportation map $T$, one seeks a joint distribution (referred to as coupling) $\pi$ on $\X\times\X$, having marginals $\mu_0$ and $\mu_1$ along the two coordinates. The Kantorovich formulation is
\[
    \inf_{\pi \in \Pi(\mu_0, \mu_1)} \int_{\X\times\X} \|x-y\|^2 \dd\pi(x,y)
\]
where $\Pi(\mu_0, \mu_1)$ is the space of all couplings with the marginals $\mu_0$ and $\mu_1$. 
In case the optimal transport map exists, the optimal coupling coincides with $\hat{\pi} = ({\rm Id} \times \hat{T})_\#\mu_0$,
where ${\rm Id}$ denotes the identity map. 

The square root of the quadratic transportation cost provides a metric on  $\PD$, known as the Wasserstein $2$-metric and denoted by $W_2$, which makes $\PD$ a geoesic space and induces a formal Riemannian structure on $\PD$ as discussed in \cite{ambrosio2004gradient,villani2008optimal}. Specifically, a constant-speed geodesic between $\mu_0$ and $\mu_1$ is given by
\begin{equation}
    \mu_t=\{(1-t)x+t\hat{T}(x)\}_\#\mu_0,~~0\leq t \leq 1
\end{equation}
 and is known as {displacement interpolation} or a {McCann  geodesic}; i.e., it satisfies
\[
    W_2(\mu_s,\mu_t)=(t-s)W_2(\mu_0,\mu_1),\quad 0\leq s<t\leq 1.
\]
In the Kantorovich formulation, the geodesic reads
\begin{equation}
    \mu_t=\{(1-t)x+ty\}_\#\hat{\pi},~~0\leq t \leq 1.
\end{equation}

We recall the definition of weak convergence of probability measures:
A sequence $\{\mu_k\}_{k\in \mathbb{N}}\subset \PD$ converges weakly to $\mu$, written as $\mu_k\rightharpoonup \mu$, if
\[\lim_{k\to \infty} \int_\X f(x)\dd \mu_k = \int_\X f(x)\dd \mu \quad{\rm for~all}~ f\in C_b(\X),\]
where $C_b(\X)$ is the Banach space of continuous, bounded and real-valued functions on $\X$.
\begin{lemma}[Gluing lemma~\cite{ambrosio2004gradient,villani2008optimal}]
\label{gluing}
Let $\X_1$, $\X_2$, and $\X_3$ be three copies of $\X$. 
Given three probability measures $\mu_i(x_i)\in \mathcal{P}_2(\mathbb{X}_i),~i=1,2,3$ and the couplings $\pi_{12}\in \Pi(\mu_1,\mu_2)$, and $\pi_{13}\in\Pi(\mu_1,\mu_3)$, there exists a probability measure $\pi(x_1,x_2,x_3)\in \mathcal{P}_2(\X_1\times \X_2 \times \X_3)$ such that $(x_1,x_2)_\#\pi=\pi_{12}$ and $(x_1,x_3)_\#\pi=\pi_{13}$. Furthermore, the measure $\pi$ is unique if either $\pi_{12}$ or $\pi_{13}$ are induced by a transport map.
\end{lemma}

Thus, for any two given couplings, which are consistent along the shared coordinate, the gluing lemma states that we can find a multi-coupling on the product space $(\X_1\times \X_2 \times \X_3)$ whose projections onto each pair of coordinates match the given couplings. 
We now briefly discuss three subtopics of interest in sequel.

\subsection{Gaussian marginals}
In case $\mu_i \sim \mathcal N({m}_i, {C}_i)$ for $i\in\{0,1\}$ are Gaussian with mean $\mu_i$ and covariance $C_i$, respectively, the solution to OMT can be given in closed form~\cite{malago2018wasserstein}
\begin{equation}
\label{Gaussian_W2}
    W_2(\mu_0,\mu_1)=\sqrt{|| {m}_0- {m}_1||^2+{\rm tr}( {C}_0+ {C}_1-2  {S})}
\end{equation}
where ${\rm tr}(.)$ stands for trace and $ {S}$ is an optimal (uniquely defined) cross-covariance term which turns out to be
\begin{equation}
     S= (C_0C_1)^{\frac{1}{2}} ={C}_0^{1/2}( {C}_0^{1/2} {C}_1 {C}_0^{1/2})^{1/2} {C}_0^{-1/2}.
\end{equation}
The McCann  geodesic $\mu_t$ for all $0\leq t \leq 1$
is a Gaussian distribution with mean $ {m}_t = (1-t) {m}_0 +t {m}_1$ and covariance
\begin{equation}
\label{Gaussian_geodesic}
     {C}_{t}= {C}_0^{-1/2}((1-t) {C}_0+t( {C}_0^{1/2} {C}_1 {C}_0^{1/2})^{1/2})^2  {C}_0^{-1/2}.
\end{equation}

\subsection{Discrete measures}
Suppose the marginals are discrete probability measures on a finite set $X\subset \Rd$, that is, $\mu_0 =\sum_{x_0\in X}p_{ x_0}\delta_{ x_0}$ and $\mu_1 =\sum_{ x_1\in X}q_{ x_1}\delta_{ x_1}$, where the non-negative weights $p_{ x_0}$ and $q_{ x_1}$ are such that $\sum_{x_0\in X}p_{ x_0}=\sum_{ x_1\in X}q_{ x_1}=1$. The transport plan is now in the form of a matrix $( {\Pi}_{ x_0, x_1})_{( x_0, x_1)\in X\times X}$ and its entries represent the amount of mass moved from $ x_0$ to $ x_1$. The Kantorovich problem in discrete setting can be written as the following linear program:
\begin{align}
\label{Opt_transport}
    \min_{ {\Pi}} \quad & \sum_{ x_0, x_1\in X}{c( x_0, x_1) {\Pi}_{ x_0, x_1}}\\ \nonumber
    {\rm s.t.} \quad &\sum_{ x_1\in X}{ {\Pi}_{ x_0, x_1}=p_{ x_0}},~ \forall  x_0 \in X\\ \nonumber
    &\sum_{ x_0\in X}{ {\Pi}_{ x_0, x_1}=q_{ x_1}},~ \forall  x_1\in X\\ \nonumber
    & {\Pi}_{ x_0, x_1}\geq 0,~ \forall  x_0, x_1\in X,
\end{align}
where $c( x_0, x_1)=|| x_0- x_1||_2^2$ is the transportation cost. (Throughout, we assume that transportation costs are quadratic.) Although cast as a linear program, this problem suffers from a heavy computational cost in large scale applications. It was pointed out in~\cite{cuturi2013sinkhorn} that there are computation advantages by introducing an entropy regularization term since, in that case, the problem can then be solved efficiently using the Sinkhorn algorithm.

\subsection{Multi-marginal optimal transportation}
In multi-marginal transport, a set of marginals are given and a law is sought that is consistent with the given marginals and minimizes a cost. This problem and its applications are surveyed in \cite{pass2015multi,nenna2016numerical}. The Kantorovich formulation of this problem for given marginals $\{\mu_i\}_{i=1}^N$ and transportation cost $c( x_1,\cdots, x_N)$ is to minimize 
\begin{equation}
\int_{\X^N} c( x_1,\cdots, x_N)d\gamma( x_1,\cdots, x_N)
\end{equation}
where the multi-coupling $\gamma \in \mathcal{P}_2(\X^N)$ is such that ${ x_i}_{\#}\gamma=\mu_i$. This is a linear optimization problem over a weakly compact and convex set for which the numerical methods to solve it efficiently are well studied in \cite{nenna2016numerical,benamou2019generalized}.

\section{Regression in Wasserstein space using measure-valued curves}\label{sec:regression}
We generalize regression problems, thought of in the setting of a Euclidean space, to the space of probability measures. To this end, for a given set $\{\mu_{t_i}\}_{i=1}^{N}\subset \PD$ of probability measures that are indexed by timestamps $\{t_i\}_{i=1}^{N} \subset \left[0,1\right]$, we seek suitable interpolating {\em measure-valued curves}. Notice that  when $\mu_{t_i}$ is absolutely continuous with respect to the Lebesgue measure, (by a slight abuse of notation) we use $\mu_{t_i}$ to denote both the measure and its density function, depending on the context.

\subsection{Measure-valued curves}

 We consider primarily two classes of functions (curves) from the time interval $[0,1]$ to the state space $\X$, linear and quadratic polynomials, denoted by
 ${\rm Lin}([0, 1],\X)$ and ${\rm Quad}([0, 1],\X)$, respectively.
 Generically, we use $\Omega$ to denote either class. 
In the sequel, we consider probability laws on linear, quadratic, and possibly other classes of functions, so as to build corresponding classes
 of measure-valued curves.
 
 For instance, in the case of
 $\Omega={\rm Lin}([0, 1],\X)$, a probability law can be expressed as a coupling between the endpoints of line segments, i.e., a probability law $\pi$ on $\X^2:=\X\times \X$. This is due to the fact that there is a bijective correspondence $(X_{0,1})$ between each element in $\Omega$ and $\X^2$ using the endpoints at $t=0$ and $t=1$, i.e., $x_0$ and $x_1~\in\X$, such that for any $\omega=(\omega_t)_{t\in[0,1]}\in\Omega$, we have $X_{0,1}(\omega):=(
x_0,x_1)$. We equip $\Omega$ with the canonical  $\sigma$-algebra generated by the
projection maps $(X_t)_{t\in [0,1]}$, defined by $X_t(\omega):=\omega_t$. 
In this study, we consider only the probability measures with finite second moments over $\X^2$, that is, $\mathcal{P}_2(\X^2)$, and accordingly the induced probability measures over $\Omega$. Given any probability measure $\pi$ on $\X^2$, the one-time marginals can be obtained through $\nu_t:=((1-t)x_0+tx_1)_\#\pi,~t\in[0,1]$.
 
 An alternative representation of a probability law on $\Omega={\rm Lin}([0, 1],\X)$ may be given in terms of a coupling between one endpoint, $x_0$, and a velocity $v$. In this representation, the one-time marginals are cast as $\nu_t:=(x_0+tv)_\#\pi,~t\in[0,1]$. In the rest of this paper, we use the first representation to define probability laws on ${\rm Lin}([0, 1],\X)$.

Similar setting can be defined for $\Omega={\rm Quad}([0, 1],\X)$ where $\Omega$ is, clearly, bijective to $\X^3$. Herein, any probability law on $\Omega$ can be expressed as a probability measure over $\X^3$, namely, $\pi\in \mathcal{P}_2(\X^3)$. Also, the one-time marginals can be obtained via $\nu_t:=(x_0+tx_1+t^2x_2)_\#\pi,~t\in[0,1]$. For ease of notation, we use $x_0$, $x_1$, and $x_2$ to denote the initial point, velocity, and acceleration, respectively. Although one can consider other parameterizations  of quadratic curves, e.g. through three points lying on each curve with suitable timestamps, we derive the results for the former representation without loss of generality. 
 
In the next subsection, we detail the regression formalism of minimizing in the Wasserstein sense the distance of  distributional data from respective marginals of
measure-valued linear and quadratic curves in $\PD$, namely,
\begin{equation}
\label{lin_measure_valued}\mathcal{G_{\rm Lin}}:=\{(\nu_t)_{t\in [0,1]}\subset \PD ~|~ \nu_t=((1-t) {x}_0+t {x}_1)_\#\pi, \pi\in \mathcal{P}_2(\X^2)\},\end{equation}
and
\begin{equation}
\label{quad_measure_valued}
\mathcal{G_{\rm Quad}}:=\{(\nu_t)_{t\in [0,1]} \subset \PD ~|~ \nu_t=(x_0+tx_1+t^2x_2)_\#\pi, \pi\in \mathcal{P}_2(\X^3)\}.
\end{equation}
We point out that any $(\nu_t)_{t\in [0,1]}$ in $\mathcal{G_{\rm Lin}}$, or $\mathcal{G_{\rm Quad}}$, is absolutely continuous~\cite[Theorem 1]{karimi2020statistical}, which amounts to the fact that the metric derivative~\cite{ambrosio2008gradient}
\[
|\nu'|(t):= \lim_{s\rightarrow t}{\frac{W_2(\nu_s, \nu_t)}{|s-t|}}\leq m(t)
\]
is bounded by some function $m(t) \in L^1(0,1)$ for almost all $t\in\left(0,1\right)$.
 
\subsection{Regression problems}

 Regression analysis seeks to model the relationship between variables, which in our case are probability measures. We consider time as the independent variable and, thereby, regression in the space of probability measures amounts to identifying a flow of one-time marginals which may capture possible underlying dynamics.

Thus, given a set of ``points'' $\{\mu_{t_i}\}_{i=1}^{N}\subset \PD$, we pose the regression problem
\begin{equation}
\label{eq:primal}
\inf_{\nu\in\mathcal G}\;\sum_{i=1}^{N}\lambda_i W_2^2(\nu_{t_i},\mu_{t_i}),
\end{equation}
where  $\mathcal G$ is either $\mathcal{G_{\rm Lin}}$ or $\mathcal{G_{\rm Quad}}$, and the ``weights'' $\lambda_i>0$ ($i=1,\cdots,N$) satisfy $\sum_{i=1}^N\lambda_i=1$. 
 
Linear measure-valued curves represent linear-in-time flows which advance an initial probability measure at $t=0$ to  another one at $t=1$, and generate correlations across the time interval. Conversely, these linear curves are specified by correlation of their end points, and therefore, problem \eqref{eq:primal} becomes one of minimizing over $\pi\in\mathcal{P}_2(\X^2)$ that represents the coupling between the marginals at $t=0$ and $t=1$. Specifically, \eqref{eq:primal} can be cast as 
\begin{equation}
\label{linear_prob}
\inf_{\pi\in \mathcal{P}_2(\X^2)} F_1(\pi):=\sum_{i=1}^{N}\lambda_i W_2^2(((1-t_i) {x}_0+t_i {x}_1)_\#\pi,\mu_{t_i}).
\end{equation}
In \eqref{linear_prob} we  assume $N\geq 3$ since, trivially, for $N=2$ any coupling between the two endpoints results in a zero cost.

\begin{remark}
It is important to contrast \eqref{linear_prob} with the geodesic regression problem \cite{jiang2011geometric,bigot2017geodesic,cazelles2018geodesic,wang2013linear} that seeks a geodesic in Wasserstein space to likewise approximate the distributional data $\mu_{t_i}$ ($i\in\{1,\ldots,N\}$). To this end, note that a curve $\nu_t=((1-t) {x}_0+t {x}_1)_\#\pi$ is a Wasserstein geodesic when $\pi$ is an {\em optimal} coupling between two marginals (typically, the end-point ones); the space of such optimal couplings is a strict subset of $\mathcal{P}_2(\X^2)$. Thus, the formulation \eqref{linear_prob}
is a relaxation of the geodesic regression in a way that may be seen as analogous to Kantorovich's relaxation of Monge's problem. Our motivation stems from the computational complexity of geodesic regression rooted in the fact that $F_1(\pi)$ is not displacement convex (see \cite[Section III]{karimi2020statistical}). In contrast, 
in the next section, we will see that \eqref{linear_prob} can be recast as a multi-marginal transport problem and solved efficiently using Sinkhorn's algorithm.
\end{remark}

Analogously, we define the regression problem for measure-valued quadratic curves by minimizing \eqref{eq:primal} over $(\nu_t)_{t\in [0,1]}\in \mathcal{G_{\rm Quad}}$, leading to
\begin{equation}
\label{quadratic_prob}
\inf_{\pi\in \mathcal{P}_2(\X^3)} F_2(\pi):=\sum_{i=1}^{N}\lambda_i W_2^2((x_0+tx_1+t^2x_2)_\#\pi,\mu_{t_i}).
\end{equation}
In this case, the hypothesis class consists of flows which are quadratic in time. It may represent distributions of inertial (mass) particles moving in the space according to quadratic functions in time under the influence of a conservative force field. As mentioned earlier, the regression formalism can be generalized to any other hypothesis classes, e.g. higher-order curves (cubic, quartic), sinusoids with variable amplitudes and frequencies, and so on. In the present work, however, we restrict our attention to linear and quadratic measure-valued curves.

The existence of minimizers is stated next.
\begin{proposition}
\label{solution_exist}
Problems \eqref{linear_prob} and \eqref{quadratic_prob} have minimizing solutions. 
\end{proposition}
\begin{proof}
The proof follows from Proposition 2.3 in \cite{agueh2011barycenters}. For completeness, we detail the steps of the proof for \eqref{linear_prob}; the proof of \eqref{quadratic_prob} follows similarly.

Let $\left\{\pi_n\right\}_{n=1}^\infty$ be a minimizing sequence of \eqref{linear_prob}. Since the data $\mu_{t_i}\in\PD$ ($i\in\{1,\ldots,N\}$), the sequence $\left\{\int_{\X^2} \|x\|_2^2 \dd \pi_n\right\}_{n=1}^\infty$ remains bounded. This implies that $\left\{\pi_n\right\}_{n=1}^\infty$ is tight. Therefore, Prokhorov’s theorem guarantees the existence of a sub-sequence weakly converging to some $\pi^*\in \mathcal{P}_2(\X^2)$. The lower semi-continuity of 
Wasserstein distance shows that $F_1(\pi)$ in \eqref{linear_prob} is a lower semi-continuous functional. As a result, $F_1(\pi^*)\leq \liminf\limits_{n\rightarrow \infty} F_1(\pi_n)={\displaystyle \inf_{\pi} F_1(\pi)}$. This proves that \eqref{linear_prob} has a minimizer. 
\end{proof}

Our next proposition states that \eqref{linear_prob} and \eqref{quadratic_prob} behave well with respect to scaling time. It is stated for problem \eqref{linear_prob} and highlights the fact that changing the units of time does not affect the solution.

\begin{proposition}
\label{scallable_time}

Suppose for given $\{\mu_{t_i}\}_{i=1}^{N}\subset \PD$, with $\{t_i\}_{i=1}^{N} \subset \left[0,T\right]$, $\hat{\pi}^T\in \mathcal{P}_2(\X^2)$ is a minimizer of
\begin{equation}
\label{scaled_problem}
\inf_{\pi\in \mathcal{P}_2(\X^2)} \sum_{i=1}^{N}\lambda_i W_2^2(((T-t_i) {x}_0+t_i {x}_1)_\#\pi,\mu_{t_i}).
\end{equation}
Then, $\hat{\pi}^1:=(T{x_0},T{x_1})_\# \hat{\pi}^T$  is a minimizer of \eqref{linear_prob} for $\{\frac{t_i}{T}\}_{i=1}^{N} \subset \left[0,1\right]$.
\end{proposition}
\begin{proof}
For each term in \eqref{scaled_problem}, let $\hat{\eta}_i(x_0,x_1,y) \in \Pi(\hat{\pi}^T,\mu_{t_i})$ be such that $((T-t_i) {x}_0+t_i {x}_1,y)_\#\hat{\eta}_i$ is an optimal coupling between its marginals. Such $\hat{\eta}_i$ exists due to Proposition 7.3.1 in \cite{ambrosio2008gradient}.
Using \eqref{push_forward}, we have
\begin{align*}
W_2^2(((T-t_i) {x}_0+t_i {x}_1)_\#\hat{\pi}^T,\mu_{t_i})&=\int_{\X^3} \|(T-t_i)x_0+t_ix_1-y\|_2^2\dd \hat{\eta}_i(x_0,x_1,y)\\
&\hspace*{-20pt}=\int_{\X^3} \|(1-\frac{t_i}{T}) (T{x_0})+\frac{t_i}{T}(Tx_1)-y\|_2^2\dd \hat{\eta}_i(x_0,x_1,y)\\
&\hspace*{-20pt}=\int_{\X^3} \|(1-\frac{t_i}{T}) {x_0}+\frac{t_i}{T}x_1-y\|_2^2\dd \left\{(Tx_0,Tx_1,y)_\#\hat{\eta}_i\right\}.
\end{align*}
It follows that $\hat{\pi}^1=(T{x_0},T{x_1})_\# \hat{\pi}^T$.
\end{proof}

The proposition above shows the regression problems behave nicely with respect to time scaling and thus, without loss of generality, we can always assume the timestamps normalized to lie within the interval $\left[0,1\right]$. Analogous steps can be carried out to show that $\hat{\pi}^1=({x_0},T{x_1},T^2{x_2})_\# \hat{\pi}^T$ is a minimizer of \eqref{quadratic_prob} for $\{\frac{t_i}{T}\}_{i=1}^{N} \subset \left[0,1\right]$ when $\hat{\pi}^T$ is a minimizer for a corresponding problem with timestamps over a window $[0,T]$ with $T>1$.

\section{Multi-marginal formulation} \label{sec:multimarginal}

In this section, we show that measure-valued regression can be recast as a  multi-marginal optimal transportation problem. Numerically, this is extremely beneficial when combined with entropy regularization as described in the next section. First, we provide the result for measure-valued quadratic curves in the following. 

\begin{theorem}
\label{theorem:mm} 
Problem (\ref{quadratic_prob}) can be recast as
\begin{align}
\label{main_eq0}
 \inf_{\substack{\pi}} F_2(\pi)= &\inf_{\substack{\gamma}}
\int_{\X^{N+3}} \sum_{i=1}^{N}\lambda_i\nonumber \|x_0+t_ix_1+t_i^2x_2-y_i\|_2^2\dd\gamma(x_0,x_1,x_2,y_1,\cdots,y_N) \nonumber\\
& {  \quad {\rm s.t.}\quad {y_i}_\#\gamma=\mu_{t_i}, \forall i=1,\cdots,N},  
\end{align}
with $\gamma\in \mathcal{P}_2(\X^{(N+3)})$,
 $\pi\in \mathcal{P}_2(\X^{3})$. Moreover, a minimizer of the right-hand side ($\hat{\gamma}$) exists
 and  $\hat{\pi}=(x_0,x_1,x_2)_\#\hat{\gamma}$ is a minimizer of left-hand side.
\end{theorem}

\begin{proof}
%We sketch the key steps of the proof. % and we refer to \cite{inpreparation} for details. 
First, suppose $\pi\in \mathcal{P}_2(\X^{3})$ and $\mu_t\in \PD$ are such that $\nu_t=(x_0+tx_1+t^2x_2)_\#\pi,~t\in[0,1]$ and $\eta_t\in \Pi(\pi,\mu_t)$, namely, a coupling between $\pi$ and $\mu_t$. Define
\begin{equation*}%\label{eq:weta}
    W_{\eta_t}(\nu_{t},\mu_t):=\int_{\X^{4}} \|x_0+tx_1+t^2x_2-y\|_2^2\dd\eta_t(x_0,x_1,x_2,y).
\end{equation*}
for which we have $W_{2}^{2}(\nu_{t},\mu_t) \leq W_{\eta_t}^2(\nu_{t},\mu_t),~\forall t\in [0,1].$ We can show the tightness of this inequality for some $\hat{\eta}_t$, namely, $W_{2}^{2}(\nu_{t},\mu_t) =
W_{\hat{\eta}_t}^2(\nu_{t},\mu_t)$.
%\min_{\eta_t\in \Pi(\pi,\mu_t)} W_{\eta_t}^2(\nu_{t},\mu_t)$.
To do so, we assume that $\hat{\Lambda}_t$ is an optimal coupling between $\nu_{t}$ and $\mu_t$. Also, we define $\rho_t\in \mathcal{P}_2(\X^4)$ as
\[\rho_t:=(x_0+tx_1+t^2x_2,x_1,x_2,y)_\#\hat{\eta}_t.\]
The existence of $\hat{\eta}_t$ amounts to finding the probability measure $\rho_t$ which fulfils the following properties:
\begin{equation*}
(z_1,z_4)_\#\rho_t=\hat{\Lambda}_t \quad \text{and}\quad   (z_1,z_2,z_3)_\#\rho_t=(x_0+tx_1+t^2x_2,x_1,x_2)_\#\pi
\end{equation*}
where $(z_1,z_4)_\#\rho_t$ denotes the projection of $\rho_t$ onto the product space of first and last coordinates, and $(z_1,z_2,z_3)_\#\rho_t$ is its projection onto the product space of the first three coordinates. 
Since the projections of $\hat{\Lambda}_t$ and $(x_0+tx_1+t^2x_2,x_1,x_2)_\#\pi$ onto their first coordinates are consistent, i.e., equal to $\nu_{t}$, by the application of Gluing Lemma (Lemma \ref{gluing}), we conclude the existence of $\rho_t$. Moreover, as the map $(x_0+tx_1+t^2x_2,x_1,x_2,y)$ is invertible, $\hat{\eta}_t$ exists as well.

Using the disintegration theorem~\cite[Theorem 5.3.1]{ambrosio2008gradient}, we can extend this result to a family of measures $\left\{\mu_{t_i}\right\}_{i=1}^N \subset \PD$ to show that for given $\pi\in \mathcal{P}_2(\X^{3})$,
\begin{equation*}
%\label{Prop_5_eq}
\sum_{i=1}^N \lambda_i W_{2}^{2}(\nu_{t_i},\mu_{t_i}) =\min_{\substack{\gamma\in \mathcal{P}_2(\X^{N+3})\\
{y_i}_\#\gamma=\mu_{t_i}\\
(x_0,x_1,x_2)_\#\gamma=\pi}}
\sum_{i=1}^N \lambda_i W_{\gamma}^2(\nu_{t_i},\mu_{t_i}) \nonumber \\
\end{equation*}
A minimizer of problem above ($\hat{\gamma}$) can be constructed as
\begin{equation}
\label{gamma_construction}
\dd\hat{\gamma}(x_0,x_1,,x_2,y_1,\ldots,y_N)=\dd(\hat{\eta}^{x_0,x_1,x_2}_{t_1} \times \ldots \times \hat{\eta}^{x_0,x_1,x_2}_{t_N}) (y_1,\ldots,y_N) \dd\pi(x_0,x_1,x_2).
\end{equation}
In \eqref{gamma_construction}, the disintegration of each measure $\hat{\eta}_{t_i}$ is written as $\dd\hat{\eta}_{t_i}(x_0,x_1,x_2,y_i)=\dd \hat{\eta}^{x_0,x_1,x_2}_{t_i}(y_i)\dd\pi(x_0,x_1,x_2)$.

According to Proposition \ref{solution_exist}, the minimizer $\hat{\pi}$ of left-hand side in \eqref{main_eq0}  exists. Thereby, using \eqref{gamma_construction}, we can obtain a minimizer of the multi-marginal formulation in \eqref{main_eq0} ($\hat{\gamma}$). This proves existence of a solution for our multi-marginal formulation and also,  $\hat{\pi}=(x_0,x_1,x_2)_\#\hat{\gamma}$. The proof is complete. 
\end{proof}

\begin{comment}
\textbf{Remark 2}: The constraint in the multi-marginal formulation guarantees that $\gamma$ has marginals equal to the observations, i.e., ${y_i}_\#\gamma=\mu_{t_i}$. Moreover, it
has a probabilistic interpretation in the context of conditionally independent  random variables. First, one can define the random vectors $\{ {Y}_i\}_{i=1}^N$ on $\Rd$ associated with distributions $\{{y_i}_\#\gamma\}_{i=1}^N$, respectively. In the similar way, $ {X}$ is defined as a random vector on $\mathbf{R}^{2d}$ associated with distribution $(\x_0,\x_1)_\#\gamma$. Then, the constraint on the right-hand side of (\ref{main_eq}) guarantees that each pair of random variables $ {Y}_i$ and $ {Y}_j$ $(i\neq j)$ are conditionally independent given~$ {X}$. 
\end{comment}

Similarly, a multi-marginal formulation for \eqref{linear_prob} is provided in the following corollary. The proof is skipped as it resembles that of Theorem \ref{theorem:mm}.
\begin{corollary}
Problem (\ref{linear_prob}) can be recast as
\begin{align}
\label{main_eq}
 \inf_{\substack{\pi}} F_1(\pi)= &\inf_{\substack \gamma}
\int_{\X^{N+2}} \sum_{i=1}^{N}\lambda_i\nonumber \|(1-t_i)x_0+t_ix_1-y_i\|_2^2\dd\gamma(x_0,x_1,y_1,\cdots,y_N) \nonumber\\
& { \quad  {\rm s.t.}\quad {y_i}_\#\gamma=\mu_{t_i}, \forall i=1,\cdots,N},  
\end{align}
with $\gamma\in \mathcal{P}_2(\X^{(N+2)})$,
 $\pi\in \mathcal{P}_2(\X^{2})$. Moreover, a minimizer of the right-hand side ($\hat{\gamma}$) exists
 and  $\hat{\pi}=(x_0,x_1)_\#\hat{\gamma}$ where $\hat{\pi}$ is a minimizer of left-hand side.
\end{corollary}

The following proposition shows the  consistency of our method with regression in Euclidean space when the target distributions are Dirac measures. %This proposition is phrased for the linear curves, however, similar statement holds for quadratic curves. 
\begin{proposition}
\label{Euclid_geod_consistency}

If all the observations are Dirac measures, i.e., $\mu_{t_i}=\delta_{{v}_i},~i=1,\cdots,N$ where $\left\{v_i\right\}_{i=1}^N\subset \X$, we have 
\[
\inf_{\substack{\pi\in \mathcal{P}_2(\X^2)}} F_1(\pi)=
\inf_{\substack{x_0,x_1\in \X}} \sum_{i=1}^{N}\lambda_i \|(1-t_i)x_0+t_ix_1-{v}_i\|_2^2,\]
\[
\inf_{\substack{\pi\in \mathcal{P}_2(\X^3)}} F_2(\pi)=
\inf_{\substack{x_0,x_1,x_2\in \X}} \sum_{i=1}^{N}\lambda_i \|x_0+t_ix_1+t_i^2x_2-{v}_i\|_2^2.\]
\end{proposition}
\begin{proof}
See \cite[Proposition 7]{karimi2020statistical} for the proof. 
\end{proof}

In the original formulation of multi-marginal optimal transport, constraints are typically given on all
marginals. However, in \eqref{main_eq0} and \eqref{main_eq}, constraints are only imposed on a subset of marginals of multi-coupling $\gamma$. We show that these problems can be written in an original formalism of multi-marginal transportation.  The following proposition provides this result for linear curves; similar argument holds for quadratic curves. 
\begin{proposition}
\label{original_mm}
For every $y=(y_1,\cdots,y_N) \in \X^N$ define 
\begin{equation*}
    (\hat{x}_0(y),\hat{x}_1(y))=\underset{(x_0,x_1)\in \X^2}{\operatorname{\rm arg~ min}} \sum_{i=1}^N\lambda_i\|(1-t_i)x_0+t_ix_1-y_i\|_2^2 
\end{equation*}
which is a well-defined map from $\X^N$ to $\X^2$ (since the linear regression in Euclidean space has a unique solution in a closed form). 
Then, we have
\begin{align}
\label{second_mmf}
 \inf_{\substack {\pi}} F_1(\pi)= &\inf_{\substack{\gamma'\in\mathcal{P}_2(\X^N)}}
\int_{\X^N} \sum_{i=1}^{N}\lambda_i\nonumber \|(1-t_i)\hat{x}_0(y)+t_i\hat{x}_1(y)-y_i\|_2^2\dd\gamma'(y_1,\cdots,y_N) \nonumber\\
&  \quad {\rm s.t.}\quad {y_i}_\#\gamma'=\mu_{t_i}, \forall i=1,\cdots,N,  
\end{align}
and also, $\hat{\pi}=(\hat{x}_0(y),\hat{x}_1(y))_\#\hat{\gamma}'$ where $\hat{\pi}$ and $\hat{\gamma}'$ are  minimizers of the left- and right-hand sides, respectively. 
\end{proposition}
\begin{proof}
The proof is straightforward by noticing that for any $\gamma\in\mathcal{P}_2(\X^{N+2})$ which respects all the constraints on the marginals, we have
\begin{align*}
&\int_{\X^{N+2}} \sum_{i=1}^{N}\lambda_i \|(1-t_i){x}_0+t_i{x}_1-y_i\|_2^2\dd\gamma(x_0,x_1,y_1,\cdots,y_N) \\
&\geq \int_{\X^{(N+2)}} \left(\inf_{z_0,z_1\in\X} \sum_{i=1}^{N}\lambda_i \|(1-t_i){z}_0+t_i{z}_1-y_i\|_2^2\right) \dd\gamma(x_0,x_1,y_1,\cdots,y_N)\\
&=\int_{\X^{(N+2)}}  \sum_{i=1}^{N}\lambda_i \|(1-t_i)\hat{x}_0(y)+t_i\hat{x}_1(y)-y_i\|_2^2 \dd\gamma(x_0,x_1,y_1,\cdots,y_N).
\end{align*}
By taking the infimum of both sides of inequality above over $\gamma\in\mathcal{P}_2(\X^{N+2})$ and using the identity $\gamma'=(y_1,\cdots,y_N)_\#\gamma$ (projection onto the last $N$ coordinates), we arrive at the result.
\end{proof}

\begin{remark}
\label{discrete_equivalence}
One implication of Proposition \ref{original_mm} is that in case all the distributional data are supported on a discrete set of points, so does $\hat{\gamma}'$. Specifically, if each $\mu_{t_i}$ is supported on a finite set $X_i\in \X$ for all $i=1,\cdots,N$, then ${\rm supp} (\hat{\gamma}')$, i.e., support of $\hat{\gamma}'$, lies within $X_1\times \cdots X_N$ (it is straightforward to show this result, see Proposition 7 in \cite{delon2020wasserstein}). Also, $\hat{\pi}=(\hat{x}_0(y),\hat{x}_1(y))_\#\hat{\gamma}'$ is concentrated on a finite set, that is, the projection of ${\rm supp} (\hat{\gamma}')$ under the map $(\hat{x}_0(y),\hat{x}_1(y))$. In this case, the problem of measure-valued curves admits a solution in which only a finite number of measure-valued curves in $\X$ have non-zero measures. Therefore, for discrete target measures the problem reduces to a finite-dimensional linear programming as formulated in the following section. Figure \ref{weighted_regression} illustrates this concept for three discrete measures as the distributional data at three instants of time. Two possible trajectories for a mass particle at $t_0$ are shown with dotted lines. The solid lines represent the best fitting lines for each trajectory (resulted from linear regression in Euclidean space). One can observe that comparing to the lower fitting line, the upper one leads to a smaller value for the sum of squared residuals in $\X$, as it passes closer to its three corresponding points. By solving the multi-marginal problem in \eqref{second_mmf}, a smaller probability measure (weight) is expected to be assigned to the lower fitting line to penalize its higher value for the sum of squared residuals. 

\end{remark}

\begin{figure}[htb]
	\centering
	\includegraphics[width=3in]{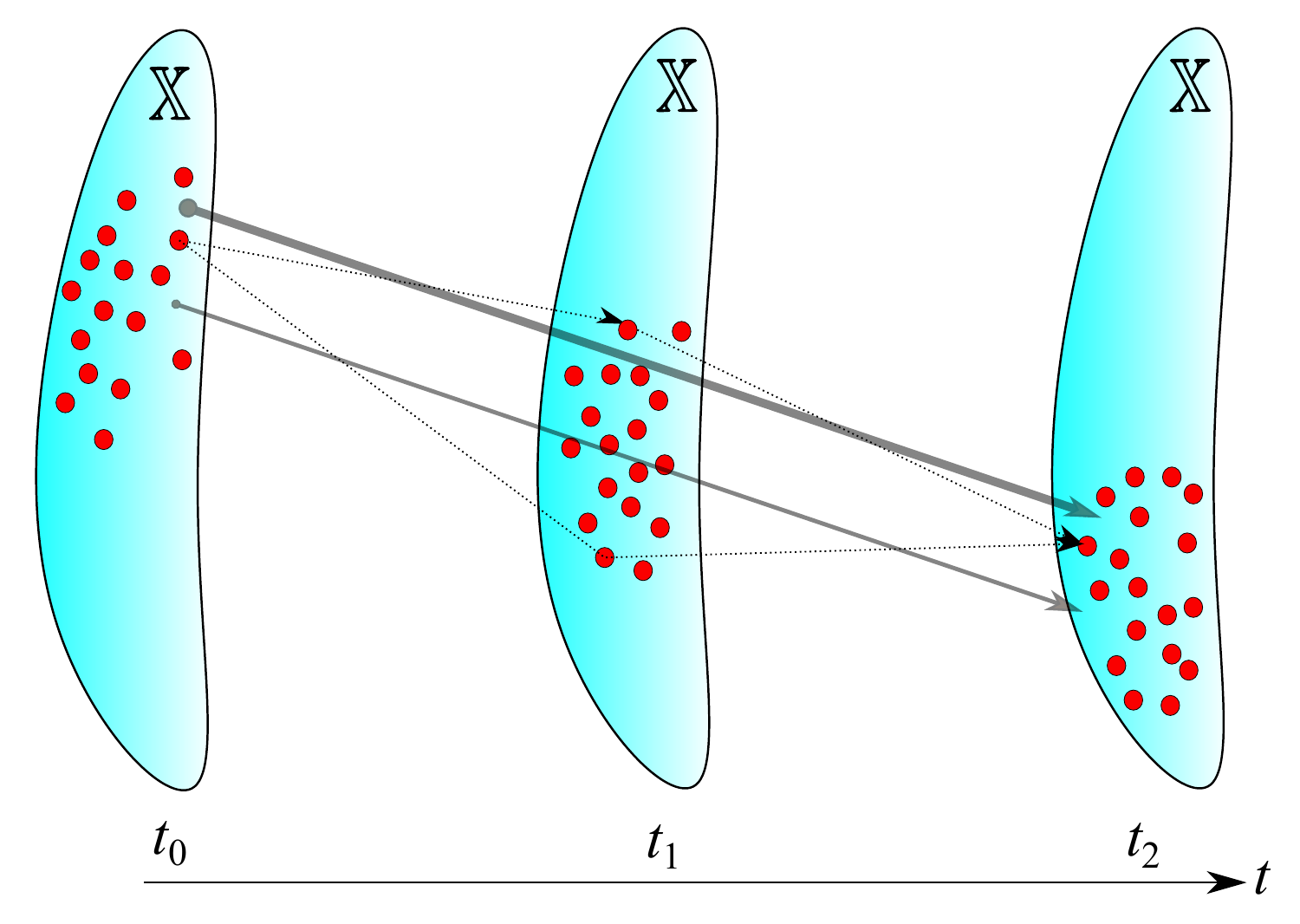}
	\caption{Illustration of measure-valued curves for discrete one-time marginals. The dotted lines show two different trajectories for a particle starting from $t_0$. The solid lines are their corresponding fitting lines resulted from linear regression in $\X$. The sum of squared residuals of the fitting line in the top has a lower value than that of the other one. The solution of multi-marginal problem assigns a higher probability measure (weight) to this fitting line. The thickness of  lines is proportional to the likelihood of each line.}
	\label{weighted_regression}
\end{figure}

\section{Discretization}
\label{sec:discretization}

In this section, we propose a strategy towards solving the multi-marginal problems introduced in the previous section. First, we express a discretized version of the problem and then invoke the entropy regularization to solve our multi-marginal formulation efficiently. This is beneficial as in many practical situations we only have a set of samples available for each one-time marginal. Thus, we can approximate each distributional data with a sum of Diracs placed at the positions of the available samples.  

\subsection{Discrete multi-marginal formulation}
Suppose for a finite set $X\subset\X$, $\mu_{t_i}=\sum_{y\in X} p^{t_i}_{y} \delta_{y}$, $i=1,\cdots,N$ are the given observations, where for each $i$ the non-negatives weights $p^{t_i}_{y}$ sum up to 1. Without of loss of generality,  it is assumed that all $\mu_{t_i}$s are supported on $X$ or a subset of it.
We define the multi-marginal problem as seeking a multi-dimensional array $({\Gamma}_{x_0,x_1,y_1,\cdots,y_N})_{(x_0,x_1,y_1,\cdots,y_N)\in X^{N+2}}$ with non-negative real elements which solves the following linear programming problem,
\begin{align}
\label{discrete}
&\min_{\substack \Gamma\geq 0 } \quad \sum_{x_0,x_1,y_1,\cdots,y_N\in X} c(x_0,x_1,y_1,\cdots,y_N)  {\Gamma}_{x_0,x_1,y_1,\cdots,y_N}\\ \nonumber
&{\rm s.t.}\quad P_{y_j}(\Gamma)=p^{t_j}_{y_j},~\forall y_j\in X,~j=1,\cdots,N
\end{align}
where,
\[c(x_0,x_1,y_1,\cdots,y_N)=\sum_{i=1}^N \lambda_i||(1-t_i)x_0+t_ix_1-y_i||^2
\]
is the cost of transport and
\begin{equation}
\label{linear_projection}
P_{y_j}(\Gamma)=\sum_{x_0,x_1,y_1,\cdots,y_{j-1},y_{j+1},\cdots,y_N\in X}  {\Gamma}_{x_0,x_1,y_1,\cdots,y_{j-1},y_j,y_{j+1},\cdots,y_N}
\end{equation}
is the projection operator on the marginal of $\Gamma$ associated with $y_j$.

Notice that $ {\Gamma}$ is analogous to the multi-coupling $\gamma$ in \eqref{main_eq}. Comparing to the definition of $\pi$ in continuous setting, we have ${\Pi}_{x_0,x_1}=P_{x_0,x_1}(\Gamma), ~\forall x_0,x_1\in X$ as the projection of multi-dimensional array ${\Gamma}$ onto $(x_0,x_1)$ obtained by summing over all the remaining entries. This leads to a probability measure over the space of linear functions represented by the endpoints in $X$.

\begin{remark}
Linear programming \eqref{discrete} is equivalent to \eqref{main_eq} if $X$ is chosen rich enough to contain ${\rm supp}(\hat{\pi})$ as described in Remark \ref{discrete_equivalence}. This assumption is not required if, instead of \eqref{main_eq}, we write the discrete version of \eqref{second_mmf} (see Remark \ref{discrete_equivalence}). However, as explained in the next subsection, we continue with the discrete formulation in \eqref{main_eq} due to the structure of its transportation cost which entails a lower time and space complexities in order to implement Sinkhorn's algorithm. 
\end{remark}

The previous formalism deals with the case of measure-valued lines in discrete setting. A similar formalism for quadratic curves seeks a multi-dimensional array $\Gamma$, such that
$({\Gamma}_{x_0,x_1,x_2,y_1,\cdots,y_N})_{(x_0,x_1,x_2,y_1,\cdots,y_N)\in X^{N+3}}$, with non-negative elements which solves
\begin{align}
\label{discrete2}
&\min_{\substack \Gamma\geq 0 } \quad \sum_{x_0,x_1,x_2,\{y_i\}_{i=1}^N\in X} c(x_0,x_1,x_2,y_1,\cdots,y_N)  {\Gamma}_{x_0,x_1,x_2,y_1,\cdots,y_N}\\ \nonumber
&{\rm s.t.}\quad P_{y_j}(\Gamma)=p^{t_j}_{y_j},~\forall y_j\in X,~j=1,\cdots,N
\end{align}
where,
\[c(x_0,x_1,x_2,y_1,\cdots,y_N)=\sum_{i=1}^N \lambda_i||x_0+t_ix_1+t^2x_2-y_i||^2,\]
and
\[
P_{y_j}(\Gamma)=\sum_{x_0,x_1,x_2,y_1,\cdots,y_{j-1},y_{j+1},\cdots,y_N\in X}  {\Gamma}_{x_0,x_1,x_2,y_1,\cdots,y_{j-1},y_j,y_{j+1},\cdots,y_N}.
\]

\subsection{Entropy regularization}
The linear programming problems in \eqref{discrete} and \eqref{discrete2} suffer from a high computational burden. However, the more efficient Sinkhorn iteration  can be employed to converge to the optimal solution of  entropy-regularized problem as explained in the following.

Given two discrete probability measures $\mu=\sum_{x\in \X} p_{x} \delta_{x}$ and $\nu=\sum_{x\in \X} q_{x} \delta_{x}$, supported on a finite set $X\subset \X$, the relative entropy (Kullback-Leibler divergence) of $\mu$ with respect to $\nu$ \cite{CovTho} is defined as
\[
 H(\mu|\nu) =  
  \begin{cases} 
   \sum_{x\in \X}p_{x} \log \frac{p_{x}}{q_{x}}& \text{if}  ~\mu\ll \nu  \\
   +\infty       & \text{otherwise,}
  \end{cases}
\]
where $\mu\ll \nu$ indicates that $\mu$ is absolutely continuous with respect to $\nu$ and $0\log 0$ is defined to be 0. Also, define
\[H(\mu):=H(\mu|1)=\sum_{x\in \X} p_{x}\log p_{x},\]
which is effectively the negative of entropy of $\mu$.

In the rest of this section, we present the results for measure-valued lines, however, one can readily state analogous results for quadratic curves.  
The entropy regularized version of our multi-marginal formulation is the convex problem
\begin{align}
\label{entr_discrete}
\min_{\substack \Gamma\geq 0 } \quad &\sum_{x_0,x_1,y_1,\cdots,y_N\in X} c(x_0,x_1,y_1,\cdots,y_N)  {\Gamma}_{x_0,x_1,y_1,\cdots,y_N}+\epsilon H({\Gamma})\\ \nonumber
{\rm s.t.}\quad & P_{y_j}(\Gamma)=p^{t_j}_{y_j},~\forall y_j\in X,~j=1,\cdots,N,
\end{align}
where $\epsilon > 0$ is a regularization parameter.

There are effective strategies to solve entropy regularized optimal transport problems, for instance, the alternating projection method (iterative Bergman projections~\cite{benamou2015iterative,benamou2016numerical, bauschke2000dykstras}), which is based on projecting sequentially an initial $\Gamma$ onto the subset corresponding to each marginal constraint.

Sinkhorn's algorithm~\cite{cuturi2013sinkhorn} is another approach which enjoys a slightly better performance in terms of space complexity and parallel computation as discussed in detail in~\cite{benamou2015iterative}. 
In this method, the optimal solution $\hat{\Gamma}$
is expressed in terms of the Lagrange dual variables, which may be computed by Sinkhorn iterations.
In the following, we first briefly touch upon this method, then by presenting a similar idea to that used in \cite{haasler2020multi}, we explain how to improve the performance of this algorithm in terms of time and space complexities.

%In the rest of this section, suppose $\exp(\cdot)$, $\log(\cdot)$, $\odot$ and $./$ denote the
%element-wise exponential, logarithm, multiplication, and division of arrays, respectively. Also, $\otimes$ denotes the outer product.

It can be shown \cite{haasler2020multi} that, for any $x_0,x_1,y_1,\cdots,y_N\in X$, the minimizer of \eqref{entr_discrete} is of the form 
\begin{equation}
\label{sinkhorn_minimizer}
   {\hat{\Gamma}}_{x_0,x_1,y_1,\cdots,y_N}=\exp(-\frac{c(x_0,x_1,y_1,\cdots,y_N)}{\epsilon})\times a^{t_1}_{y_1}\times\cdots \times a^{t_N}_{y_N},
\end{equation}
for suitable values of $a^{t_j}_{y_j},~j=1,\cdots,N$. These are dual variables in the dual problem (see e.g. \cite{nenna2016numerical}). 
In Sinkhorn's algorithm, the $a^{t_j}_{y_j}$'s in \eqref{sinkhorn} can be found by iteratively updating their values via
\begin{equation}
\label{iterations}
    a^{t_j}_{y_j}\leftarrow a^{t_j}_{y_j} \times p^{t_j}_{y_j}/P_{y_j}(\hat{\Gamma}), \forall j=1,\cdots,N, y_j\in X.
\end{equation}

It is known that in the scheme above, the sequence converges at least linearly to a minimizer of \eqref{entr_discrete} (see e.g. \cite{haasler2020multi,bauschke2000dykstras}).  

The computational drawback of Sinkhorn's algorithm lies in computing
the projections $P_{y_j}(\hat{\Gamma})$ in \eqref{iterations}, as these grow exponentially in the number of snapshots ($N$). Furthermore, a large amount of memory is required to store the array $\hat{\Gamma}$ at each iteration which leads to a space complexity issue. 
However, the specific structure of the cost in \eqref{sinkhorn_minimizer} can be exploited to mitigate the aforementioned bottlenecks. Similar ideas have been advanced in \cite{elvander2020multi,haasler2020multi}. 

Notice that we can partially decouple the cost as 
\[c(x_0,x_1,y_1,\cdots,y_N)= \sum_{i=1}^N \lambda_i c_i(x_0,x_1,y_i)\]
where,
\[
c_i(x_0,x_1,y_i)=||(1-t_i)x_0+t_ix_1-y_i||^2.
\]
The first implication of this decoupling is that, it is now not needed to store all the elements of $c(x_0,x_1,y_1,\cdots,y_N)$, but only those required to calculate $c_i$s. Moreover, the minimizer in \eqref{sinkhorn_minimizer}, can be decoupled as 
\begin{equation}
\label{sinkhorn}
   {\hat{\Gamma}}_{x_0,x_1,y_1,\cdots,y_N}=\prod_{i=1}^N a^{t_j}_{y_i}\exp(-\frac{c_i(x_0,x_1,y_i)}{\epsilon}).
\end{equation}
In the following, we explain how to leverage this structure to calculate $P_{y_1}(\hat{\Gamma})$ more efficiently. The same procedure can be utilized to compute other projections, i.e.,  $P_{y_j}(\hat{\Gamma})$, $j=2,\cdots,N$.
One can easily observe that $P_{y_1}(\hat{\Gamma})$ for fixed $x_0,x_1 \in X$,  reads
\begin{equation}
\label{reduction_Sinkhorn}
P_{y_1|x_0,x_1}(\hat{\Gamma})=a^{t_1}_{y_1}\exp(-\frac{c_1(x_0,x_1,y_1)}{\epsilon}) \prod_{i=2}^N \left( \sum_{y_i\in X} a^{t_j}_{y_i}\exp(-\frac{c_i(x_0,x_1,y_i)}{\epsilon})\right),
\end{equation}
for any $y_1\in X$, and hence,
\begin{equation}
\label{reduction_Sinkhorn2}    
P_{y_1}(\hat{\Gamma})= \sum_{x_0,x_1\in X} P_{y_1|x_0,x_1}(\hat{\Gamma}).
\end{equation}

The benefit of this approach is that the term
\[
\prod_{i=2}^N \left( \sum_{y_i\in X} a^{t_j}_{y_i}\exp(-\frac{c_i(x_0,x_1,y_i)}{\epsilon})\right)
\]
in \eqref{reduction_Sinkhorn} is independent of $y_1$ and thus it is the same for all $y_1 \in X$. The complexity of computing this term for all $x_0,x_1 \in X$ is $\mathcal{O}((N-1)|X|^{3})$, where $|X|$ is the cardinality of the discrete set $X$. This leads to $\mathcal{O}(N|X|^{3})$ as the total computational complexity of each Sinkhorn iteration by using \eqref{reduction_Sinkhorn2} to compute the projections.
Notice that computing the projections $P_{y_j}(\hat{\Gamma})$ by summing over
all the indices $x_0,x_1,y_1,\cdots,y_{j-1},y_{j+1},\cdots,y_N$
as defined in \eqref{linear_projection} scales exponentially in the value of $N$, i.e., the computational complexity of one Sinkhorn update in \eqref{iterations} is $\mathcal{O}(|X|^{N+2})$. Therefore, leveraging the structure of cost in our multi-marginal formulation, decreases the computational complexity of the Sinkhorn iterations substantially.

\section{Gaussian case}\label{sec:Gaussian}
Suppose the data are Gaussian distributions $\mu_{t_i} \sim N(0,C_{y_i})$, $i=1,\cdots,N$, where the $C_{y_i}$'s are symmetric and positive definite matrices. The means of distributions are assumed to be zero for simplicity and without loss of generality. This is due to the fact that for Gaussian measures, the means can be treated separately via ordinary regression in Euclidean space and thereby, the means for the optimal curve in $(\PD,W_2)$ can be computed as a function of $t$. 
%As a result, the problem reduces to a regression for covariance matrices. 

In practical settings where for each marginal only a set of samples is available, we can approximate each $C_{y_i}$ with the sample covariance.
The following proposition recasts \eqref{main_eq0} as a Semi-Definite Programming (SDP).
\begin{proposition}
Consider  $\mu_{t_i}\sim N(0, {C}_{y_i})$, i.e., Gaussian ``points''. A minimizing $\hat{\gamma}$ in \eqref{main_eq0}
 is Gaussian with zero mean and covariance of the form
 \begin{equation}
%\label{Covariance_Matrix}
{C}_{\gamma}=
\left[
    \begin{smallmatrix}
     {C}_{x_0}       &  {S}_{{x_0}{x_1}} & {S}_{{x_0}{x_2}} & {S}_{{x_0}{y_1}} & \dots &  {S}_{{x_0}{y_N}} \\
     {S}_{{x_0}{x_1}}^T       &  {C}_{x_1} & {S}_{{x_1}{x_2}} & {S}_{{x_1}{y_1}} & \dots &  {S}_{{x_1}{y_N}} \\
     {S}_{{x_0}{x_2}}^T & {S}_{{x_1}{x_2}}^T        &  {C}_{x_2} &  {S}_{{x_2}{y_1}} & \dots &  {S}_{{x_2}{y_N}} \\
     {S}_{{x_0}{y_1}}^T       &  {S}_{{x_1}{y_1}}^T & {S}_{{x_2}{y_1}}^T &  {C}_{y_1} & \dots &  {S}_{{y_1}{y_N}} \\
    \vdots & \vdots & \vdots & \vdots & \ddots \\
     {S}_{{x_0}{y_N}}^T       &  {S}_{{x_1}{y_N}}^T & {S}_{{x_2}{y_N}}^T &  {S}^T_{{y_1}{y_N}} & \dots &  {C}_{y_N}
\end{smallmatrix} \right]
\end{equation}
that solves
\begin{align}
\label{SDP}
   \min_{C_{\gamma}\succeq 0} 
   &\sum_{i=1}^N\lambda_i.{\rm tr}(C_{x_0}+t_i^2C_{x_1}+t_i^4C_{x_2}+C_{y_i}+2t_i S_{x_0x_1}\nonumber\\
   & +2t_i^2S_{x_0x_2}+2t_i^3S_{x_1x_2}-2S_{x_0y_i}-2t_iS_{x_1y_i}-2t_i^2S_{x_2y_i}).
\end{align}
where $C_{\gamma}\succeq 0$ indicates that $C_{\gamma}$ is positive semi-definite. 
\end{proposition}
Notice that in $C_{\gamma}$, the sub-matrices $C_{y_i}$s are given, while the other blocks are unknown. 
\begin{proof}
As the marginals ${\{\mu_{t_i}\}}_{i=1}^N$  in \eqref{main_eq0} are Gaussian and the cost function is quadratic in $x_0,x_1,x_2,y_1,\cdots,y_N$, 
it follows that $\hat{\gamma}$ in \eqref{main_eq0} is also Gaussian as in the cost and constraints only second-order moments are involved.
Simple calculation shows the quadratic cost in (\ref{main_eq0}) can be written as that in \eqref{SDP}.
%for any Gaussian $\gamma$ with covariance matrix given in (\ref{Covariance_Matrix})
%reads
%\begin{align*}
%&\scriptstyle{\int_{\mathbf{R}^{(N+2)d}} \sum_{i=1}^{N}\frac{1}{N} ||(1-t_i)x_0+t_ix_1-y_i||^2d\gamma =}\\ \nonumber
%&(\scriptstyle{1-2\overline{t}+\overline{t}^2)\text{tr}( {C}_{x_0})+\overline{t^2}\text{tr}( %{C}_{x_1})+2(\overline{t}-\overline{t^2})\text{tr}( {S}_{{x_0}{x_1}})}\\ \nonumber
%   &\scriptstyle{-\frac{2}{N}\sum_{i=1}^{N} \text{tr}((1-t_i) {S}_{{x_0}{y_i}}+t_i {S}_{{x_1}{y_i}})}.
%\end{align*}
\end{proof}

It should be noted that since $\hat{\pi}=(x_0,x_1,x_2)_\#\hat{\gamma}$, one can express the optimal curve in $\mathcal{G_{\rm Quad}}$ (defined in \eqref{quad_measure_valued}) as
\begin{equation}
\label{interp_curve_relaxed}
   \nu_t \sim N(0, C_{x_0}+t^2C_{x_1}+t^4C_{x_2}+t (S_{x_0x_1}+S^T_{x_0x_1})
    +t^2(S_{x_0x_2}+S^T_{x_0x_2})+t^3(S_{x_1x_2}+S^T_{x_1x_2})),
\end{equation}
for $t\in\left[0,1\right]$, for the optimal solution of \eqref{SDP}.

Similar results can be derived for multi-marginal formulation of measure-valued lines  in \eqref{main_eq}. In particular, 
a minimizing $\hat{\gamma}$ in \eqref{main_eq}
 is Gaussian with zero mean and covariance of the form
 \begin{equation}
%\label{Covariance_Matrix}
{C}_{\gamma}=
\left[
    \begin{smallmatrix}
     {C}_{x_0}       &  {S}_{{x_0}{x_1}}  & {S}_{{x_0}{y_1}} & \dots &  {S}_{{x_0}{y_N}} \\
     {S}_{{x_0}{x_1}}^T       &  {C}_{x_1}  & {S}_{{x_1}{y_1}} & \dots &  {S}_{{x_1}{y_N}} \\
     {S}_{{x_0}{y_1}}^T       &  {S}_{{x_1}{y_1}}^T &  {C}_{y_1} & \dots &  {S}_{{y_1}{y_N}} \\
    \vdots & \vdots  & \vdots & \ddots \\
     {S}_{{x_0}{y_N}}^T       &  {S}_{{x_1}{y_N}}^T &  {S}^T_{{y_1}{y_N}} & \dots &  {C}_{y_N}
\end{smallmatrix} \right],
\end{equation}
 that solves
\begin{align}
\label{SDP1}
   \min_{C_{\gamma}\succeq 0} 
   &\sum_{i=1}^N\lambda_i.{\rm tr}((1-t_i)^2C_{x_0}+t_i^2C_{x_1}+C_{y_i}+2t_i(1-t_i) S_{x_0x_1}\nonumber\\
   & -2(1-t_i)S_{x_0y_i}-2t_iS_{x_1y_i}).
\end{align}

%\begin{equation}
%   \scriptstyle{ \mu_t \sim N(0, (1-t)^2 {C}^*_{x_0}+t^2 %{C}^*_{x_1}+t(1-t)( {S}^*_{{x_0}{x}}+{ {S}^*_{{x_0}{x}}}^T))}
%\end{equation}

To exemplify our regression approach for Gaussian distributional data we consider a one-dimensional Ornstein–Uhlenbeck process
modeled by an Itô stochastic differential equation
\[
\dd \mathbf{X}_t=-\mathbf{X}_t\dd t+2\dd\mathbf{W}_t
\]
where $(\mathbf{W}_t)_{t\geq 0}$ is a one-dimensional standard Wiener process. Such a process models the  dynamics of an over-damped Hookean spring in the presence of thermal fluctuations. Starting from $\mathbf{X}_0=0$, the variance of $\mathbf{X}_t$ reads
\[
\sigma^2(t)=2(1-\exp(-2t)).
\]
We consider the one-time marginals of this process at 20 different timestamps starting from $t=0.1$ to $t=1$ with equal time steps. In 
practical settings where only a set of samples from each one-time marginal is available, we can approximate the Gaussian distributions using the sample means and variances.  
The SDPs in \eqref{SDP} and \eqref{SDP1}  are solved separately to obtain the optimal multi-couplings $\hat{\gamma}$ and $\hat{\pi}$ in each case. In addition, for the sake of comparison we find the best geodesic which passes as close as possible to these $20$ Gaussian marginals. This can be done easily as the marginals are one dimensional, noticing that  the geodesic  between two Gaussian distributions with standard deviations $\sigma_0$ and $\sigma_1$ is Gaussian for all $t\in\left[0,1\right]$ with standard deviation $\sigma_t=(1-t)\sigma_0+t\sigma_1$. Therefore, the geodesic regression in this setting becomes a linear regression in $\mathbb{R}^1$ seeking the values of $\sigma_0,\sigma_1>0$. Figure~\ref{Gaussian_example} illustrates the obtained curves in Wasserstein space for different values of $t$ along with the dataset. Blue curves are the target marginals. One can notice that the measure-valued quadratic curves capture the variation in the dataset better than measure-valued linear curves. Also, the geodesic regression has the poorest performance among the three. 
This ensues from the fact that in geodesic regression a curve in Wasserstein space with highest correlated endpoints is sought. However, in the framework of this paper, this constraint is relaxed which can also moderate underfitting. In Fig.~\ref{Gaussian_example}, some of the measure-valued linear or quadratic curves are represented in each sub-figure. The intensity of color is proportional to the likelihood of each path. From a  fluid mechanical point of view, this can be thought of as a flux for the mass particles. More amount of mass transports through the darker regions.

\begin{figure}[htb]  
\centering
\subfigure[Geodesic regression]{\resizebox{!}{3cm}{\includegraphics{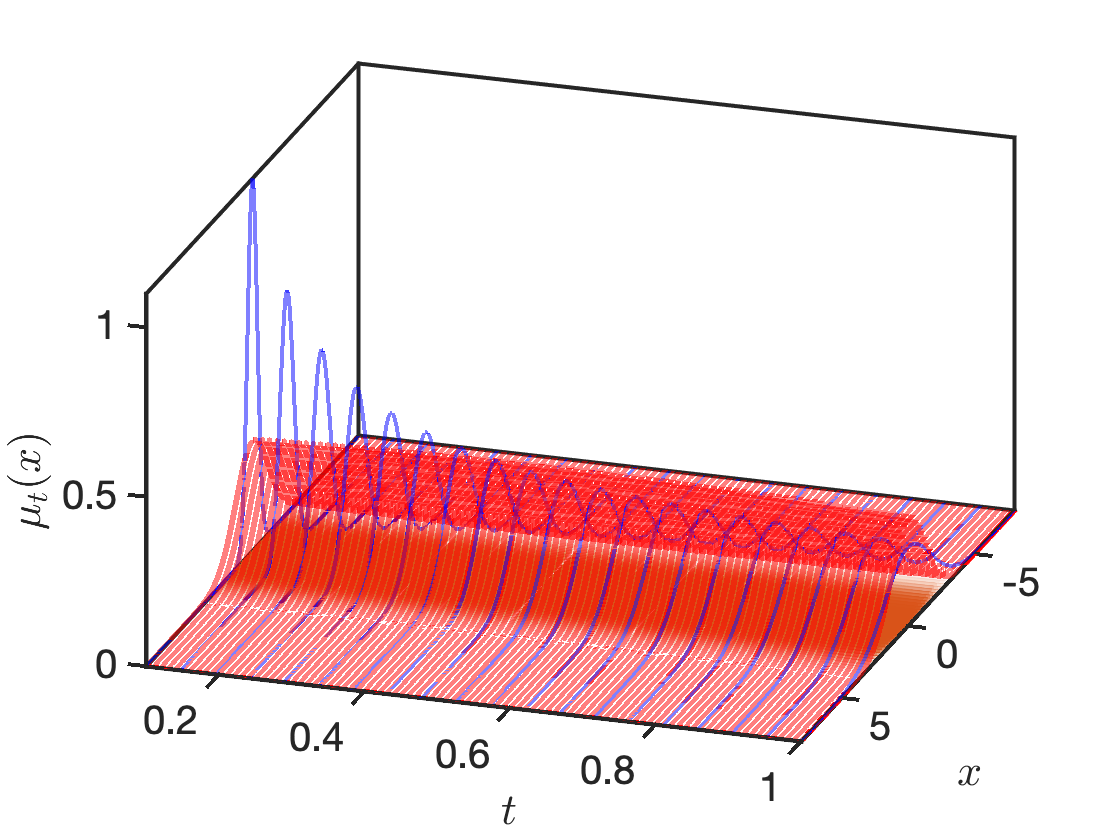}}}
\subfigure[Regression via measure-valued lines]{\resizebox{!}{3cm}{\includegraphics{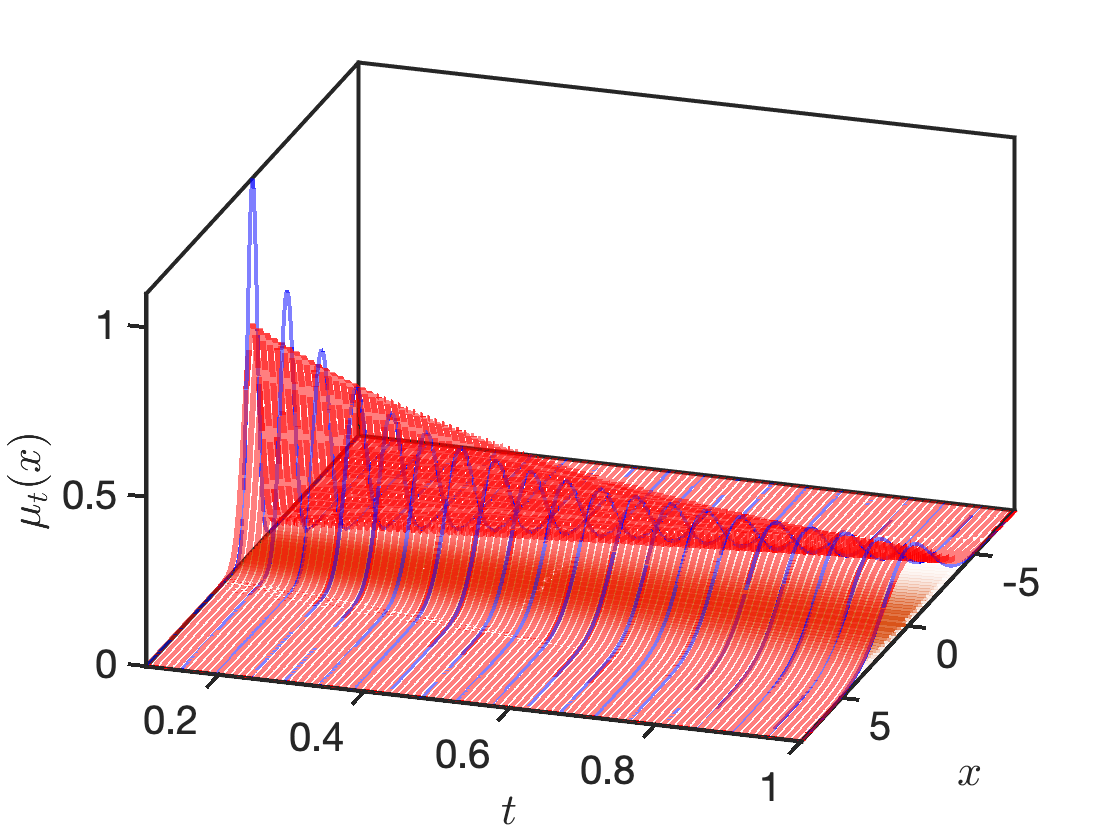}}}
\subfigure[Regression via measure-valued quadratic curves]{\label{Geo_reg}\resizebox{!}{3cm}{\includegraphics{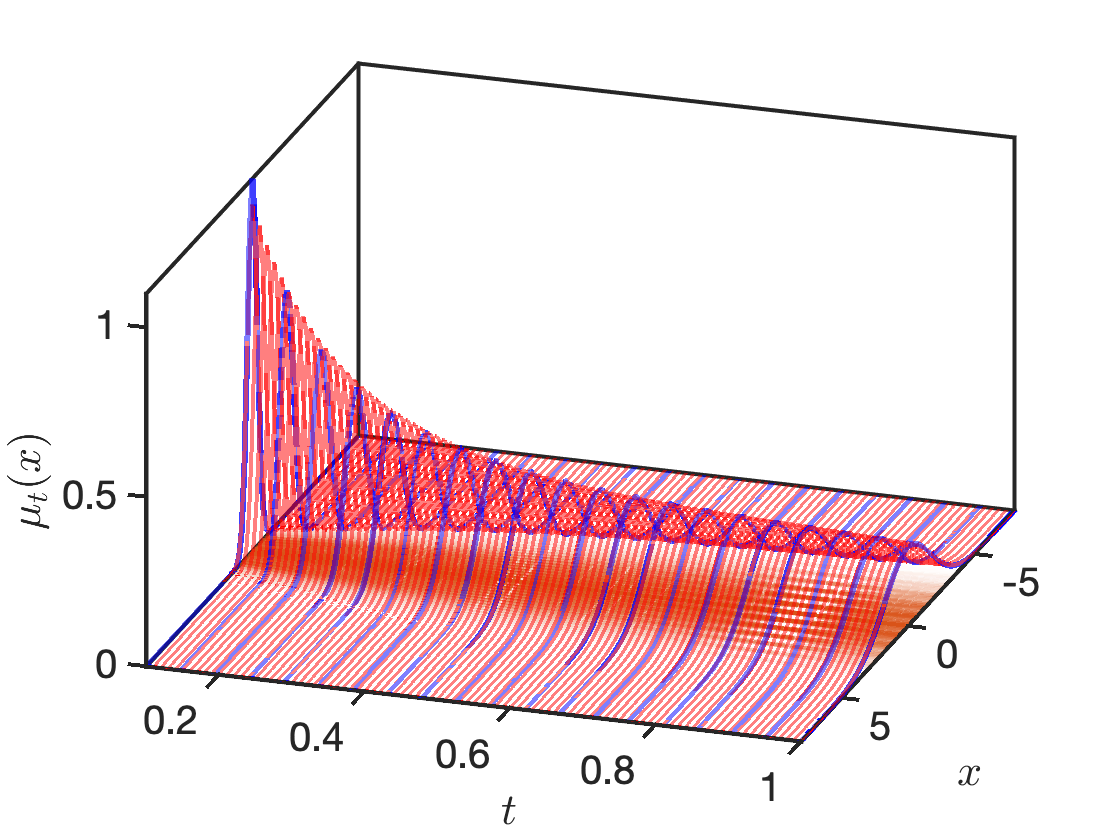}}}
\caption{Regression results for one-dimensional Gaussian marginals. Blue curves are the given distributions and red ones are the optimal curves in the Wasserstein space. The intensity of color in linear and quadratic curves is proportional to the likelihood of each path.}
\label{Gaussian_example}
\end{figure}

\section{Gaussian mixtures}\label{sec:GMM}
Linear combinations of Gaussian measures can model multi-modal  densities, which are broadly used to study properties of populations with several subgroups. More generally, the set of all finite
Gaussian mixture distributions ($\mathcal{GM}(\X)$) is a dense subset of $\PD$ in the Wasserstein metric~\cite{delon2020wasserstein}.  In fact, in principle, we can approximate any measure in $\PD$ with arbitrary precision with parameters for the Gaussian mixture determined via the Expectation-Maximization algorithm.

While the displacement interpolation of Gaussian distributions remains Gaussian, 
for Gaussian mixtures this invariance does not hold. Nevertheless, we may want to retain the Gaussian mixture structure of the interpolation due to their physical or statistical features.
In~\cite{chen2018optimal,delon2020wasserstein}, a Wasserstein-type distance on Gaussian mixture models is proposed by restricting the set of feasible coupling measures in the optimal transport problem to Gaussian mixture models. This gives rise to a geometry that inherits properties of optimal transport while it  preserves the Gaussian mixture structure. 
Specifically, for positive integers $K_0$ and $K_1$, consider the following Gaussian mixture models on $\X$,
\[\mu_0=p_{\nu_1}^0\nu_1^0+\cdots +p_{\nu_{K_0}}^0\nu_{K_0}^0,~~\mu_1=p_{\nu_1}^1\nu_1^1+\cdots +p_{\nu_{K_1}}^1\nu_{K_1}^1,\]
where each $\nu^i_j$ is a Gaussian distribution and $p^i=\left[p_{\nu_i}^i~\cdots~p_{\nu_{K_i}}^i\right]^T,~i=0,1$, are probability vectors.
Now define a Wasserstein-type distance between the two Gaussian mixtures $\mu_0$ and $\mu_1 \in \mathcal{GM}(\X)$ by minimizing 
\[\int_{\X^2} \|x-y\|_2^2~\dd \pi(x,y)\]
over $\pi\in \Pi(\mu_0,\mu_1) ~\bigcap~ \mathcal{GM}(\X^2)$. The square root of  minimum defines a metric on $\mathcal{GM}(\X)$ denoted by $W_M(\mu_0,\mu_1)$~\cite{chen2018optimal,delon2020wasserstein}. Clearly,
\[W_2(\mu_0,\mu_1) ~\leq ~ W_M(\mu_0,\mu_1),~~\forall \mu_0,\mu_1 \in \mathcal{GM}(\X).\]  

The problem above has an equivalent discrete formulation. In particular, by viewing the Gaussian mixtures as discrete probability distributions on the Wasserstein space of Gaussian distributions, we can show~\cite{delon2020wasserstein}

\begin{equation}
\label{GMM_discrete}
    W_M^2(\mu_0,\mu_1)={ \min_{w\in\Pi(p^0,p^1)}} \sum_{i,j} w_{ij} W^2_2(\nu_i^0,\nu_j^1),
\end{equation}
where $\Pi(p^0,p^1)$ denotes the space of joint distributions between the probability vectors $p^0$ and $p^1$. 
The space of Gaussian mixtures equipped with this metric is a geodesic space for which one can define the displacement interpolation~(see ~\cite{chen2018optimal,delon2020wasserstein} for further details). 

This Wasserstein-type distance between the discrete distributions on the Wasserstein space of Gaussian distributions, allows for the notion of measure-valued curves being carried over into the case of Gaussian mixtures. In other words, in the space of Gaussian distributions, the displacement interpolations (Eq. \eqref{Gaussian_geodesic}) play the role of straight lines in Euclidean space. Therefore, the goal is to find a probability measure over the space of geodesics of Gaussian distributions, for which the one-time marginals approximate a set of Gaussian mixtures indexed with timestamps. To do so, consider the set $X=\left\{ \nu_i \right\}_{i=1}^K$ which consists of a finite number of Gaussian distributions. Also, the available data \[{\left\{\mu_{t_i}=\sum_{\nu\in X} p_{\nu}^{t_i}\nu\right\}}_{i=1}^N,\] 
is a family of Gaussian mixtures, each associated with a timestamp $t_i \in \left[0,1\right]$. Each $\mu_{t_i}$ can also be thought of as a discrete probability measure over the space of Gaussian measures supported on $X$ (or a subset of $X$). By analogy with the formalism for measure-valued lines (Eq. \eqref{linear_prob}), we minimize
\begin{equation}
\label{measure_valued_mixture}
\min_{w\in \Omega} ~\sum_{i=1}^{N}\lambda_i W_M^2(~\sum_{j\ell} w_{j\ell}~g^{\nu_{j}\nu_{\ell}}_{t_i}~,~\mu_{t_i}~),
\end{equation}
where $\Omega~=~\left\{ w\in \R_+^{K\times K} ~|~\sum_{j\ell}w_{j\ell}=1 \right\}$ and $g^{\nu_{j}\nu_{\ell}}_{t}$ represents the displacement interpolation between $\nu_j$ and $\nu_\ell$.
This problem can be recast as a multi-marginal optimal transport problem which enjoys a linear structure by pursuing the same strategy introduced in Section \ref{sec:multimarginal}. In particular, \eqref{measure_valued_mixture} is equivalent to  seeking a multi-dimensional array  $({\Gamma}_{\sigma_0,\sigma_1,\nu_1,\cdots,\nu_N})_{(\sigma_0,\sigma_1,\nu_1,\cdots,\nu_N)\in X^{N+2}}$ with non-negative real elements which solves
\begin{align}
\label{GMM_Linear_MM}
&\min_{\substack \Gamma\geq 0 } \quad \sum_{\sigma_0,\sigma_1,\nu_1,\cdots,\nu_N\in X} c(\sigma_0,\sigma_1,\nu_1,\cdots,\nu_N)  {\Gamma}_{\sigma_0,\sigma_1,\nu_1,\cdots,\nu_N}\\ \nonumber
&{\rm s.t.}\quad P_{\nu_j}(\Gamma)=p^{t_j}_{\nu_j},~\forall \nu_j\in X,~j=1,\cdots,N
\end{align}
where 
\[c(\sigma_0,\sigma_1,\nu_1,\cdots,\nu_N)=\sum_{i=1}^N \lambda_iW_2^2(g^{\nu_{j}\nu_{\ell}}_{t_i},\nu_i),
\]
and $P_{\nu_j}(\Gamma)$ is the projection operator on the marginal of $\Gamma$ associated with $\nu_j$, cf. \eqref{linear_projection}. 
Also, the minimizer of \eqref{measure_valued_mixture} ($\hat{w}$) can be obtained by the projection $\hat{w}=P_{\sigma_0,\sigma_1}(\hat{\Gamma})$, where $\hat{\Gamma}$ is the minimizer of \eqref{GMM_Linear_MM}.

The formalism above is a linear programming which can be solved efficiently as one can solve the entropy regularized version of it by leveraging the generalized Sinkhorn algorithm as described in Section \ref{sec:discretization}. 

We exemplify this approach for Gaussian mixtures with the following toy example. We consider a finite set of probability measures which consists of 4 Gaussian distributions as depicted in Fig. \ref{basis}. The distributional data at 4 instants of time are constructed by choosing some probability vectors over the elements of this set. These target distributions are shown in Fig. \ref{Distributional data_GMM}. The linear programming in \eqref{GMM_Linear_MM} is solved, which results in a curve in $\PD$ for which the one-time marginals are Gaussian mixtures. The result of regression for this problem is illustrated in Fig. \ref{GMM_results} at some timestamps.  One can observe that the one-time marginals of the obtained curve capture the variation of the distributional data in time. 

\begin{figure}[htb]
	\centering
	\includegraphics[width=2.5in]{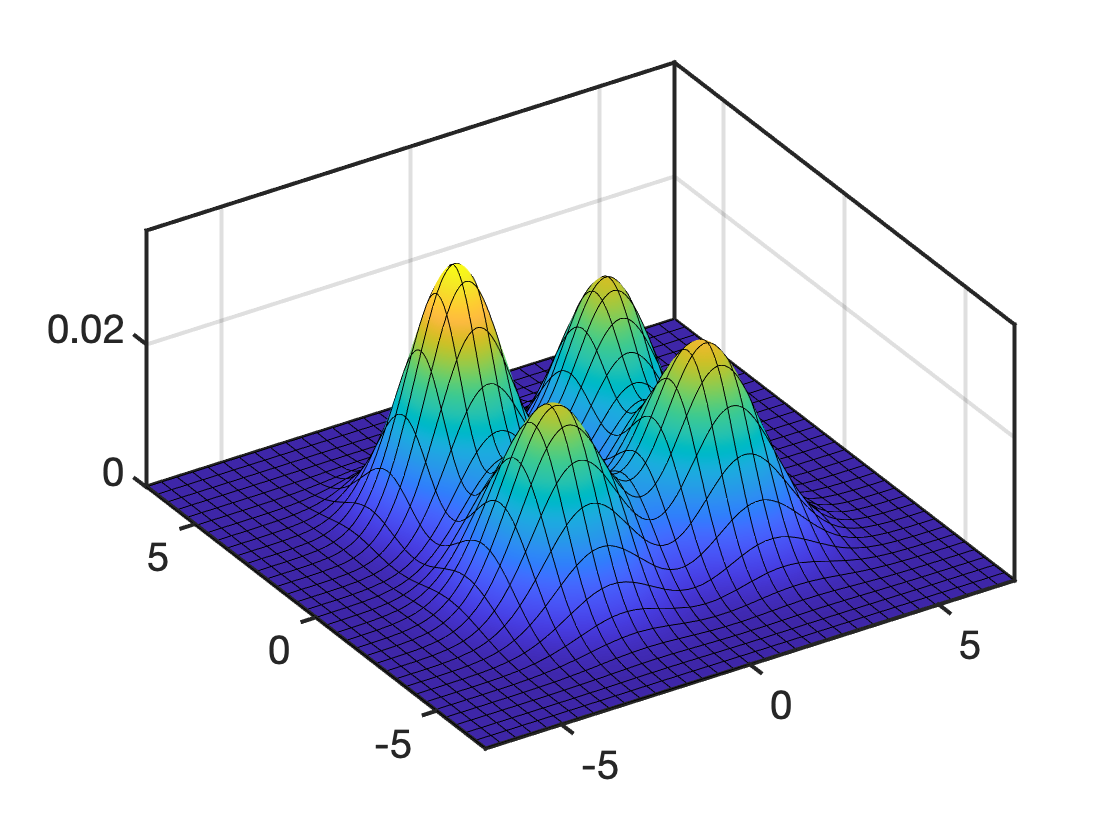}
	\caption{Gaussian Basis}
	\label{basis}
\end{figure}

\begin{figure}[htb]  
\centering
\subfigure[$t=\frac{1}{10}$]{ \resizebox{!}{2.2cm}{\includegraphics{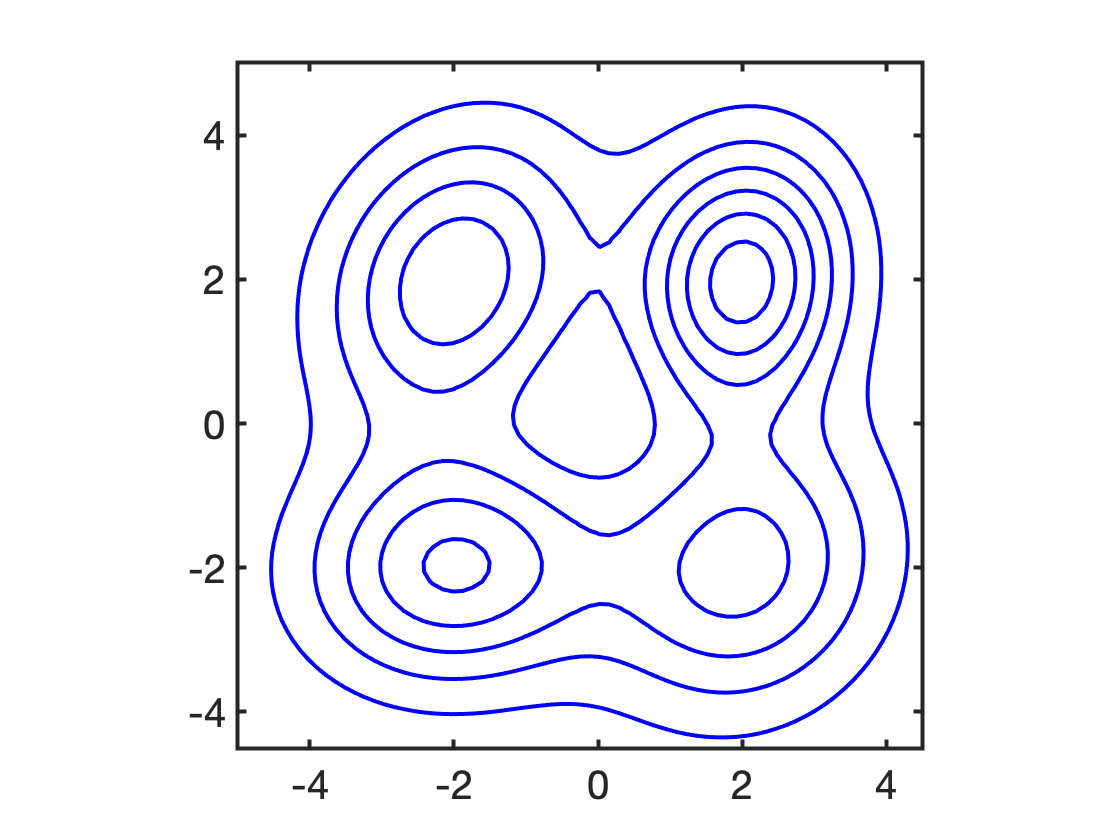}}}
\subfigure[$t=\frac{1}{3}$]{ \resizebox{!}{2.2cm}{\includegraphics{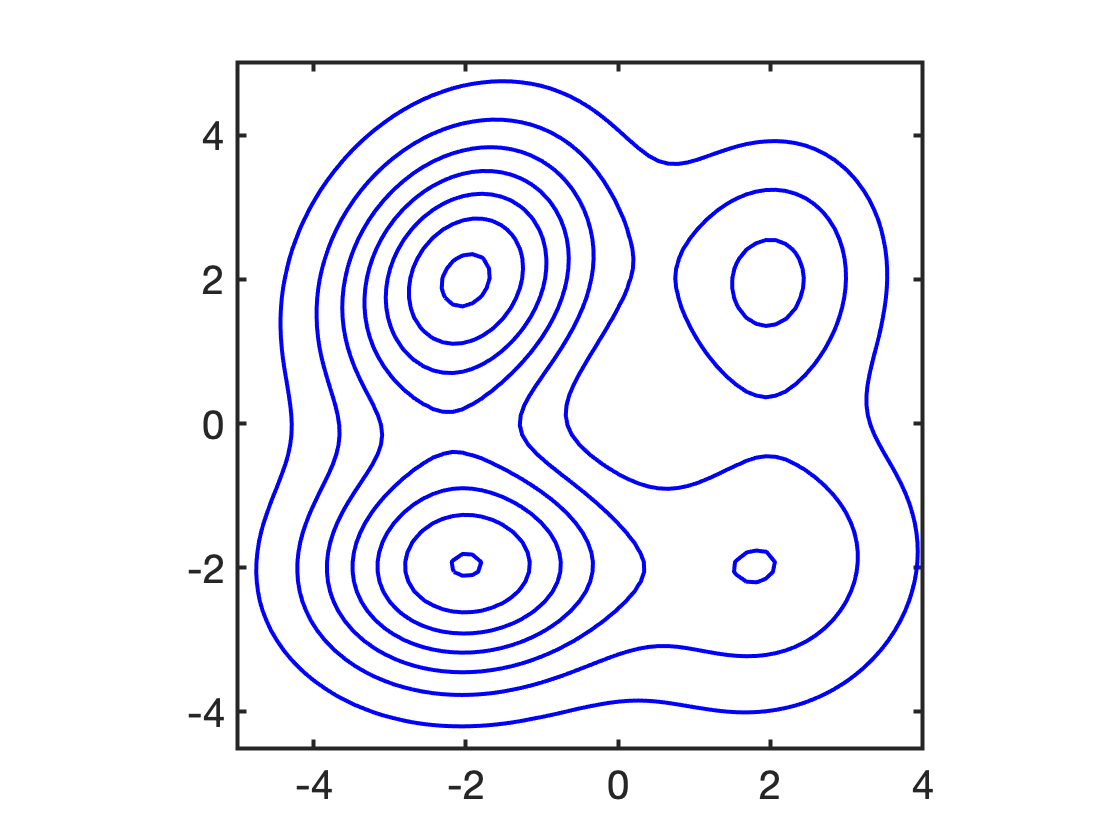}}}
\subfigure[$t=\frac{2}{3}$]{ \resizebox{!}{2.2cm}{\includegraphics{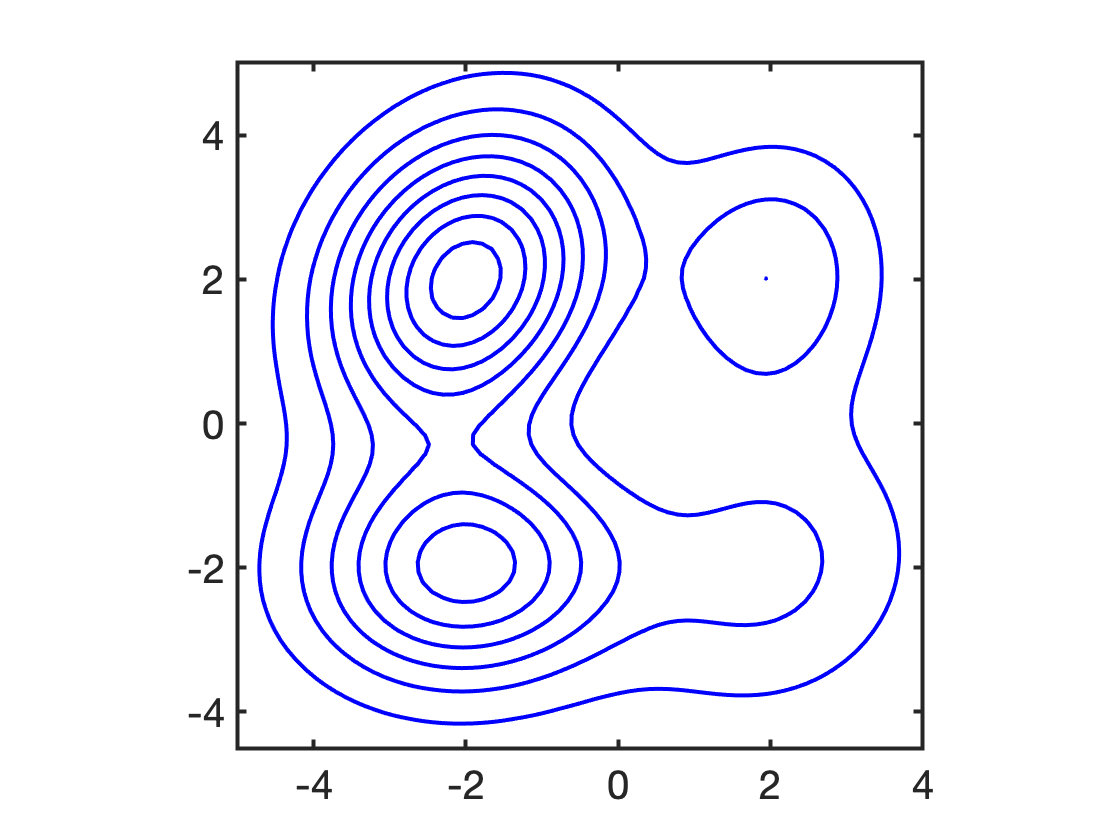}}}
\subfigure[$t=\frac{9}{10}$]{ \resizebox{!}{2.2cm}{\includegraphics{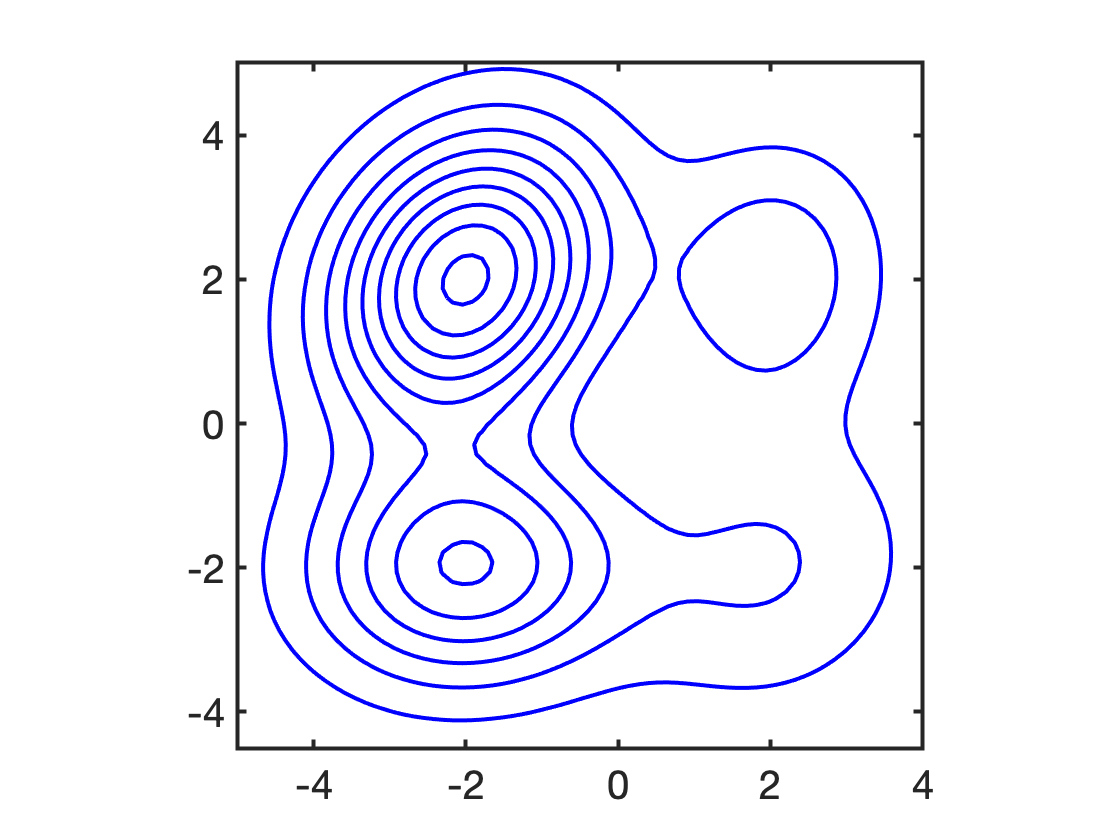}}}
\caption{Distributional data}
\label{Distributional data_GMM}
\end{figure}

\begin{figure}[htb]  
\centering
\subfigure[$t=0$]{ \resizebox{!}{2.2cm}{\includegraphics{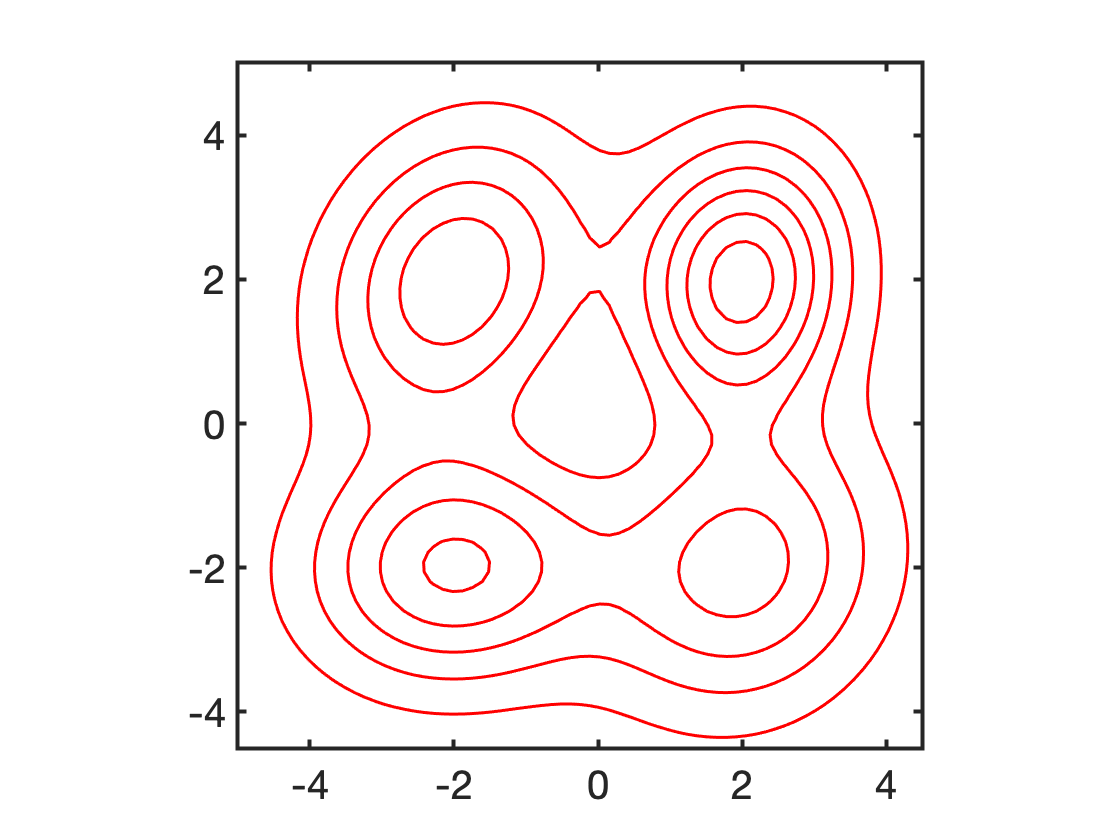}}}
\subfigure[$t=\frac{1}{10}$]{ \resizebox{!}{2.2cm}{\includegraphics{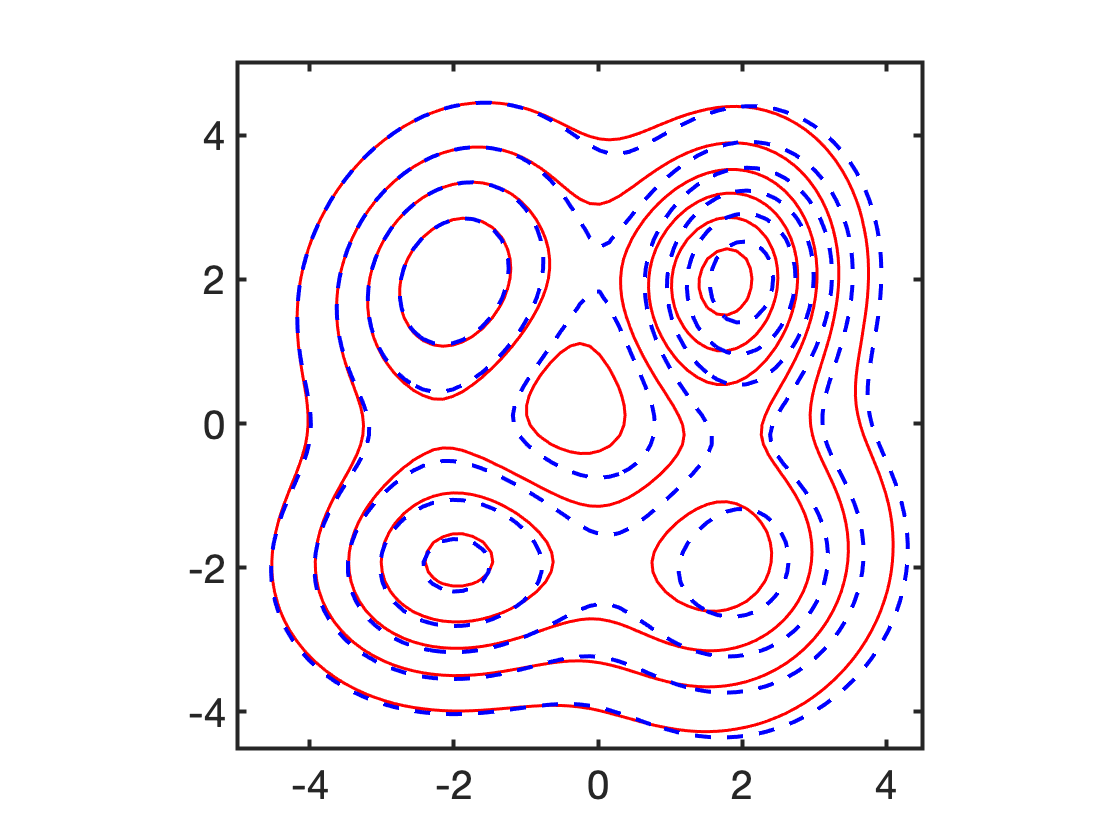}}}
\subfigure[$t=\frac{2}{10}$]{ \resizebox{!}{2.2cm}{\includegraphics{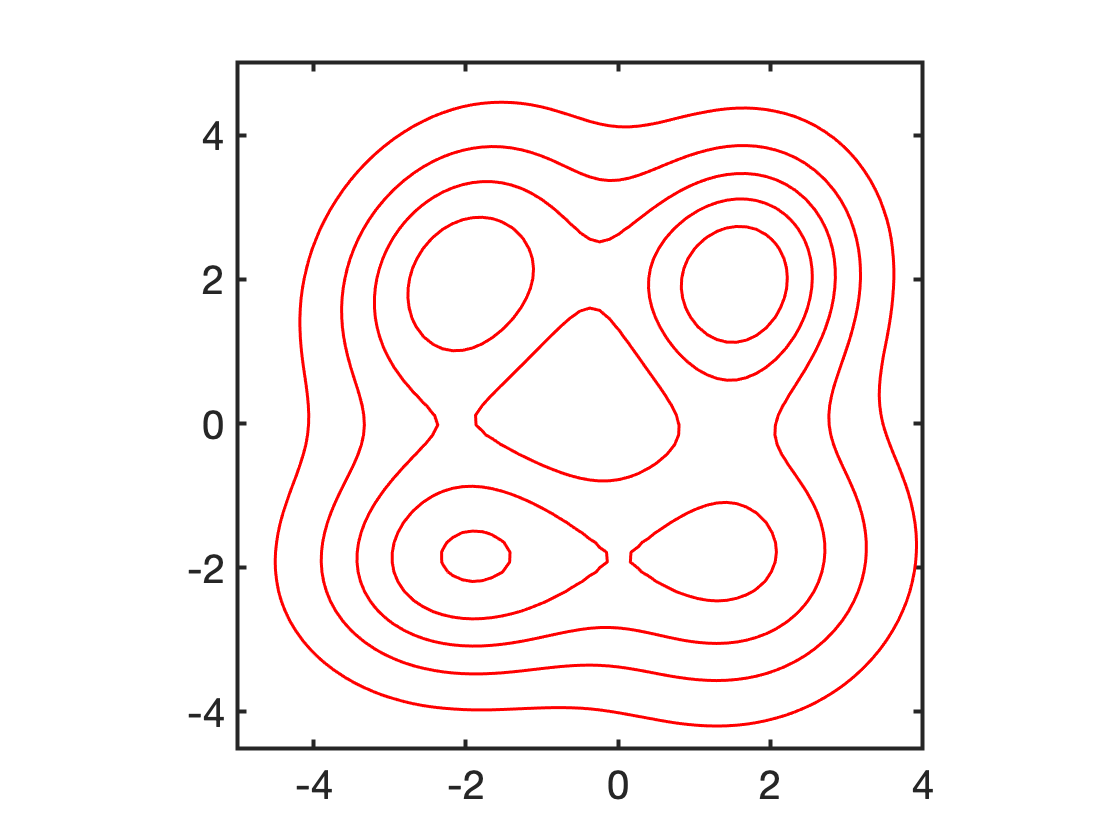}}}
\subfigure[$t=\frac{1}{3}$]{ \resizebox{!}{2.2cm}{\includegraphics{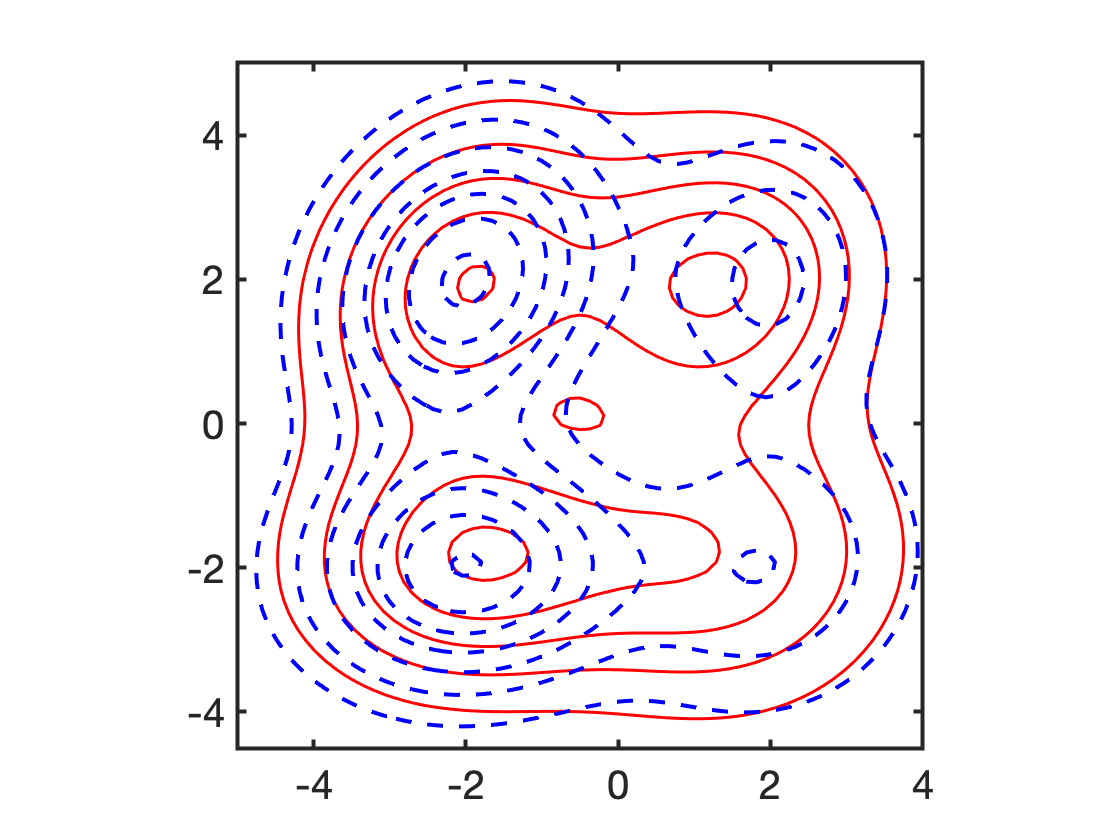}}}
\subfigure[$t=\frac{2}{3}$]{ \resizebox{!}{2.2cm}{\includegraphics{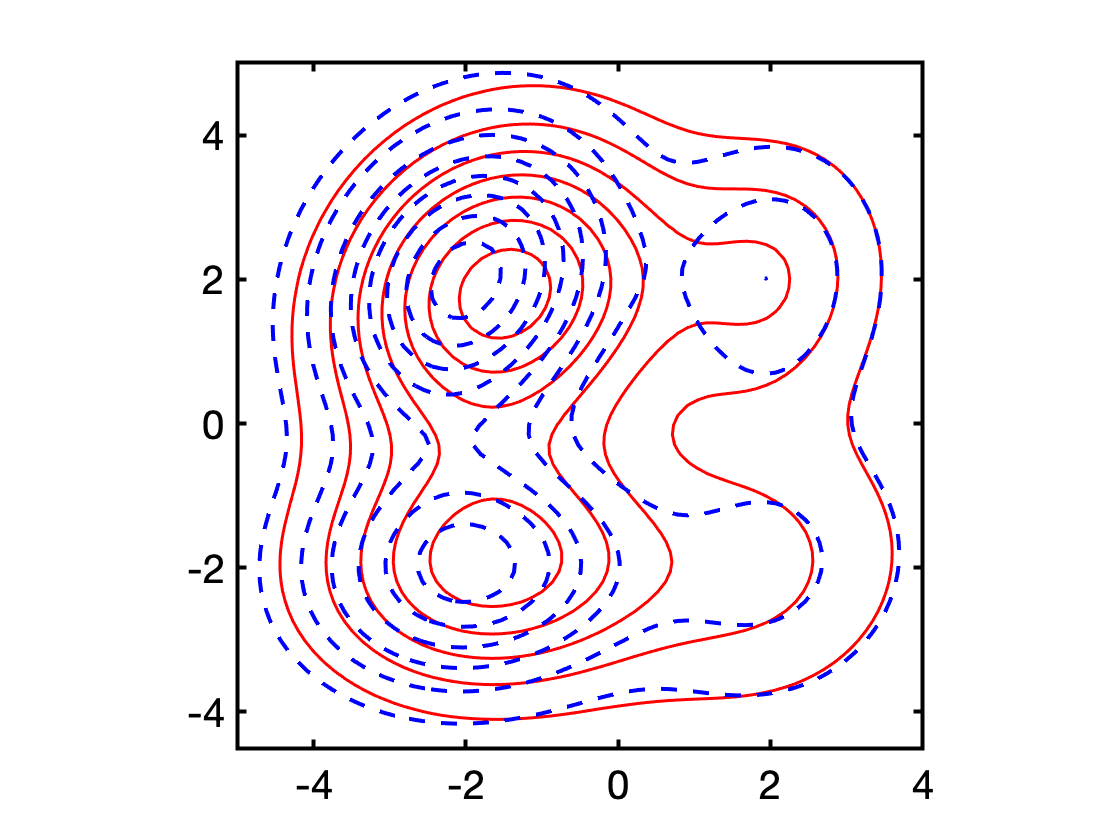}}}
\subfigure[$t=\frac{8}{10}$]{ \resizebox{!}{2.2cm}{\includegraphics{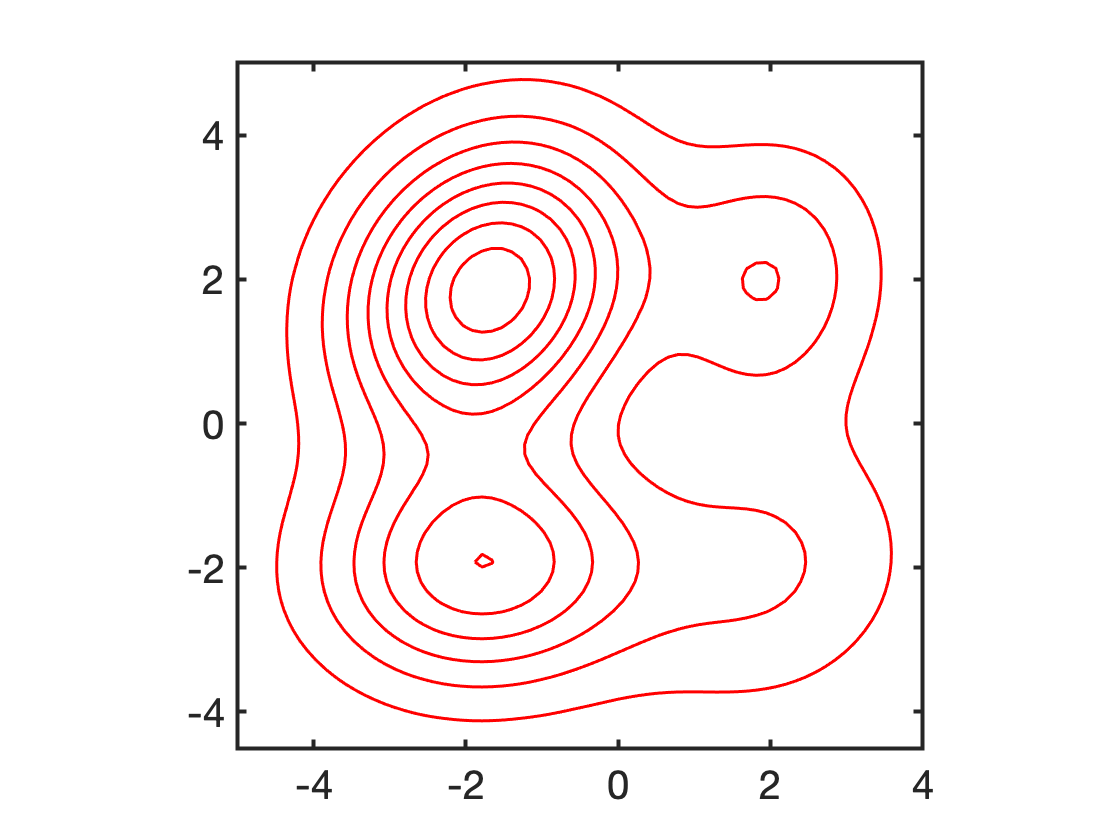}}}
\subfigure[$t=\frac{9}{10}$]{ \resizebox{!}{2.2cm}{\includegraphics{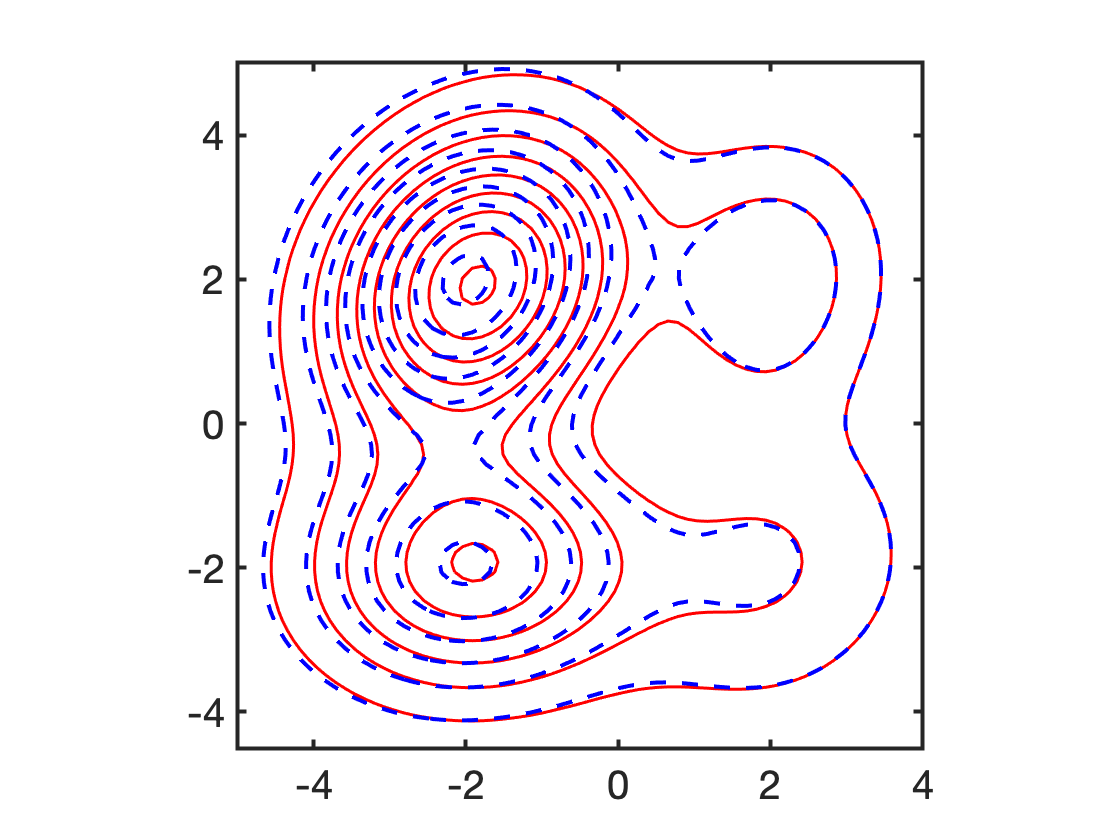}}}
\subfigure[$t=1$]{ \resizebox{!}{2.2cm}{\includegraphics{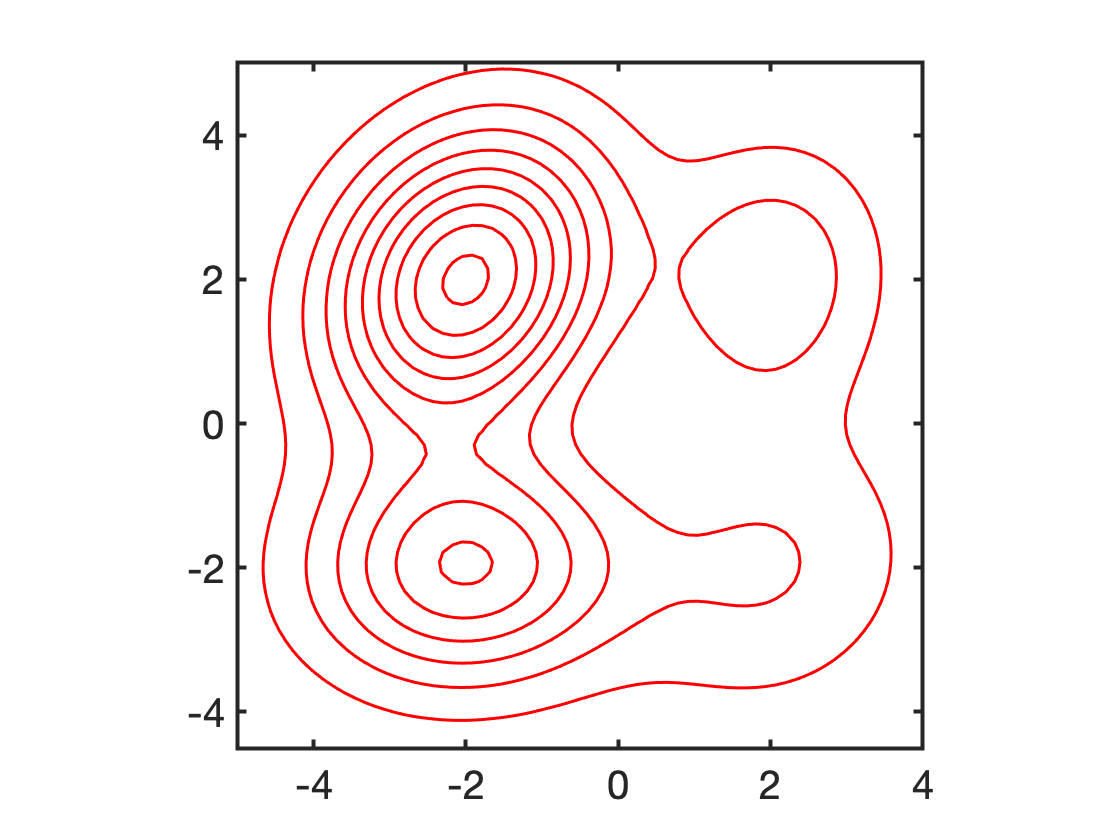}}}
\caption{The result of measure-valued geodesics regression for Gaussian mixtures.}
\label{GMM_results}
\end{figure}

\section{Estimation of invariant measures}\label{sec:invariant}
As a potential application of the proposed regression, we describe an approach to  approximate the Perron-Frobenius operator and stationary distribution (if any exists) associated with a dynamical system using a few available distributional snapshots. Most studies in the literature present numerical computation of invariant measures for known dynamics, or where the pointwise correspondence between the successive points in time is available (See \cite{korda2021convex} and references therein). In our approach, we hypothesize no information on the underling dynamics or the trajectories of particles.   

A discrete-time dynamical system 
\[
x_{k+1}=S(x_k)
\]
 on the measure space $(\X,\mathcal{B}(\X),\lambda)$ is defined by a $\lambda$-measurable state transition map
$
S: \X \to \X.
$
This map is assumed to be non-singular, which guarantees that the push-forward operator under $S$ preserves the absolute continuity of (probability) measures with respect to $\lambda$. 
In continuous-time setting, the state transition law can be represented by a flow map $x_{t+\tau}=S_\tau(x_t)$ for $\tau\geq 0$, where $x_t$ denotes the state of dynamics at time $t$.  We assume that the dynamics is time-invariant. 
The evolution of probability measures under $S_\tau$ can be written as $\mu_{t+\tau}={S_\tau}_\#{\mu_{t}}$.

Let $L^1(\mathbb{X}):=L^1(\X,\mathcal{B}(\X),\lambda)$ be the space of integrable functions on $\X$, then the \PFO ~(PFO), $P_\tau:L^1(\mathbb{X})\rightarrow L^1(\mathbb{X})$, is defined by
\begin{equation}
\label{deterministic}
\int_A P_\tau f~ d\lambda=\int_{S_\tau^{-1}(A)} f~ d\lambda,\quad \forall A\in\Sigma,
\end{equation}
for $f\in L^1(\mathbb{X})$. When $f$ is a density associated with the probability measure $\mu_f$, PFO can be thought of as a push-forward map, that is, $P_\tau\mu_f={S_\tau}_\# \mu_f$.
The connection between the dynamics and PFO can be seen in that the PFO translates the center of a Dirac measure $\delta_x\in L^1(\X)$ in compliance with the underlying dynamics, that is, ${S_\tau}_\#\delta_x=\delta_{S_\tau (x)}$. It also relates to the Koopman operator which acts on the observable functions through a duality correspondence~\cite{brunton2021modern}. 

It is standard that PFO is a Markov operator, namely, a linear operator which maps probability densities to probability densities. It is also a weak contraction (non-expansive map), in that, $\|P_\tau f\|_{L^1}\leq \|f\|_{L^1}$ for any $f \in L^1(\X)$. 
If $\mu={S_\tau}_\# \mu$, then $\mu$ is an invariant measure for $S_\tau$.
For many dynamical systems, the PFO drives the densities into an invariant one (measure, in general), which is unique if the map $S_\tau$ is ergodic with respect to $\lambda$. 

For non-deterministic dynamics where $S_\tau(x)$ is a an $\X$-valued random variable on some implicitly given probability space, the Perron-Frobenius operator reads
\begin{equation}
\label{non-deterministic}
P_\tau f(x)=\int_\X K_\tau(y,x)~f(y) ~\dd \lambda(y),
\end{equation}
where the transition density function is denoted by $K_\tau(y,x)~:~\X\times\X\rightarrow \left[0,\infty\right]$. The transition density function exists, if $S_\tau(x)$ does not assign non-zero measures to null sets~\cite{klus2018data}.

The most popular method in the literature to discretize PFO is  Ulam's method \cite{li1976finite,froyland2014computational}. In this approach, the state-space ($\X$) is divided into a finite number of disjoint measurable boxes $\{B_1, . . . , B_n\}$. The PFO is approximated with a $n\times n$ matrix with elements $p_{ij}$. To do so, first we can choose a number ($q$) of test points (samples) $\{x_l^i\}_{l=1}^q$ within each Box $B_i$ randomly. Then, the elements of this matrix can be estimated by

\[
p_{ij}=\frac{1}{q}\sum_{l=1}^{q} \mathbf{1}_{B_j}(S(x_l^i))
\]
where $\mathbf{1}_{B_j}$ denotes the indicator function for the box $B_j$.

Ulam's method requires the trajectories of test points to be available, which is not the case in many practical situations, where the trajectories of test points (i.e., mass particles, agents, or so on) are missing. We can use the method of this paper for regression, to estimate the Perron-Frobenius operator, and subsequently invariant measure corresponding to some  dynamical system based on the collective behavior of particles. In other words, we postulate no knowledge of the underlying dynamics and assume that only a limited number of one-time marginal distributions is available at different timestamps $t_i$, $i\in\{1,\ldots,N\}$. In fact, these one-time marginals are the evolution of some initial distribution at $t=0$ under the action of discretized dynamics $x_{t+\tau}=S_\tau(x_t)$. As mentioned in Proposition \ref{scallable_time}, the time can be scaled to lie within the interval $\left[0,1\right]$.
These distributions are quantized by suitable partitioning of the domain $\X=\bigcup_{\ell=1}^n B_\ell$, which is a compact set, followed by counting the particles in each of the $n$ boxes $B_\ell$ to obtain $\mu_{t_i}$, with Diracs placed at the center of each interval.
The Minimizer of multi-marginal formulation of the regression problem  (i.e., Eq. \eqref{discrete}), provides a coupling $\hat{\pi}$ between the distributions at two instants of time $t=0$ and $t=1$. Also, this can be thought of as a probability measure over the space of linear curves in $\X$, which indicates how much mass is transporting along the lines from $t=0$ to $t=1$. Putting these two views together, one can conclude that $\hat{\pi}$ gives a correlation law between the distributions of particles at $t=0$ and $t=1$, where the mass particles move at constant speeds from $t=0$ to $t=1$. It should be noted that the entropy regularization of cost can be employed to find $\hat{\pi}$ efficiently, as discussed in Section \ref{sec:discretization}. 

Notice that $\hat{\pi}$ contains the information on the distributions of the particles at  $t=0$ and $t=1$, namely, $p_{\{t=0\}}$ and $p_{\{t=1\}}$ respectively, as well as the correlation law between the two end-points. Therefore, we can determine a transition probability matrix (of a Markov chain)
\begin{equation}
\label{approximated_PFO}
Q(\ell,\ell^\prime)=\pi^*(\ell,\ell^\prime)/p_{\{t=0\}}(\ell),
\end{equation}
for $\ell,\ell^\prime\in\{1,\ldots,n\}$. From a measure-theoretic point of view $Q$ can be seen as the disintegration of $\hat{\pi}$.
This transition probability matrix can be deemed as a finite-dimensional approximation of the Perron-Frobenius operator corresponding to the underlying dynamics, that is, $P_\tau$ in either \eqref{deterministic} or \eqref{non-deterministic} for $\tau=1$. 
Assuming that the underlying dynamics is time-invariant (or time-homogeneous for non-deterministic dynamics), the invariant distribution of dynamical system can be approximated by the stationary vector of $Q$. 

To exemplify this approach, we apply it to logistic map in order to predict its asymptotic statistical properties for different values of population-growth parameter. The logistic model for population growth is
\begin{equation}
\label{logistic}
x_{k+1}=T(x_k)=rx_k(1-x_k),
\end{equation}
where $x_k\in[0,1]$, $k\in\{0,1,\ldots\}$,
and $r$ is the population-growth parameter, see~\cite{lasota2013chaos}. The behavior of dynamics changes from regular to chaotic as the parameter $r$ varies from 0 to 4.  We visualize the results for two values of $r$, namely, $r=3$ and $r=4$.

For $2\leq r \leq 3$, starting from any initial point in $(0,1)$, the population will eventually approach the same value $\frac{r-1}{r}$, so-called ``attractor". However, as $r$ approaches 3 the convergence becomes  increasingly slow.  
For $r=4$, it is known that this system displays highly chaotic behavior; in fact, starting from any initial
point $x_0\in(0,1)$, the sequence $\{x_k \,|\, k=1,2,\ldots \}$ covers densely the interval $[0, 1]$, see~\cite{ding2010statistical}. Yet, the dynamical system is statistically stable
in that, any initial probability distribution tends towards an invariant measure with density
\begin{equation}
\label{stationary_logistic}
    f_s(x)=\frac{1}{\pi\sqrt{x(1-x)}}.
\end{equation}

Our aim is to estimate where the mass particles will eventually concentrate by the iterates of logistic map using only a few probability distributions obtained from the evolution of an initial distribution under the action of this map. We do not hypothesize any information on the correlations between each pair of the probability distributions.  Namely, the logistic map is only used to construct the distributional data. To do so, the interval $\left[0,1\right]$ is partitioned into $n$ sub-intervals of equal width and the evolution of 1000 points, uniformly selected in $\left[0,1\right]$, is used to construct $N$ distributional data under the successive iterates of logistic map.  We index the data with timestamps $t_i=\frac{i-1}{N-1}$, $i=1,\cdots,N$. The logistic map herein can be though of as the flow map of a dynamical system for the time lag $\tau=\frac{1}{N-1}$.  We provide the results for different values of $N$, $n$, and the regularization parameter $\epsilon$ in Sinkhorn's algorithm, to examine the sensitivity of the results to these parameters. 

The transition probability matrix $Q$ in Eq. \eqref{approximated_PFO} is computed for different values of $N$, $n$, and $\epsilon$ and accordingly the stationary distribution of $Q$ is obtained. The results are depicted in Fig. \ref{logistic_map_r3} which show that the stationary distribution is concentrated around the stable fixed point of the logistic map ($x_{k+1}=3x_k(1-x_k)$) at $x=\frac{2}{3}$. In particular, the results for three values of $n$ are illustrated in the first row. The second row represents the impact of $N$ on estimated stationary distribution. As the number of distributional data varies from 3 to 9, we observe the stationary distribution is concentrated more densely around the fixed point. Finally, the third row relates to the sensitivity of results to regularization parameter $\epsilon$. Although for smaller values of $\epsilon$ the convergence of Sinkhorn iterates to the optimal solution becomes slower, we achieve a better result for the stationary distribution in terms of having a lower variance.   

\begin{figure}[htb]  
\centering
\subfigure[$N=5,~n=30,~\epsilon=0.1$]{ \resizebox{!}{2.9cm}{\includegraphics{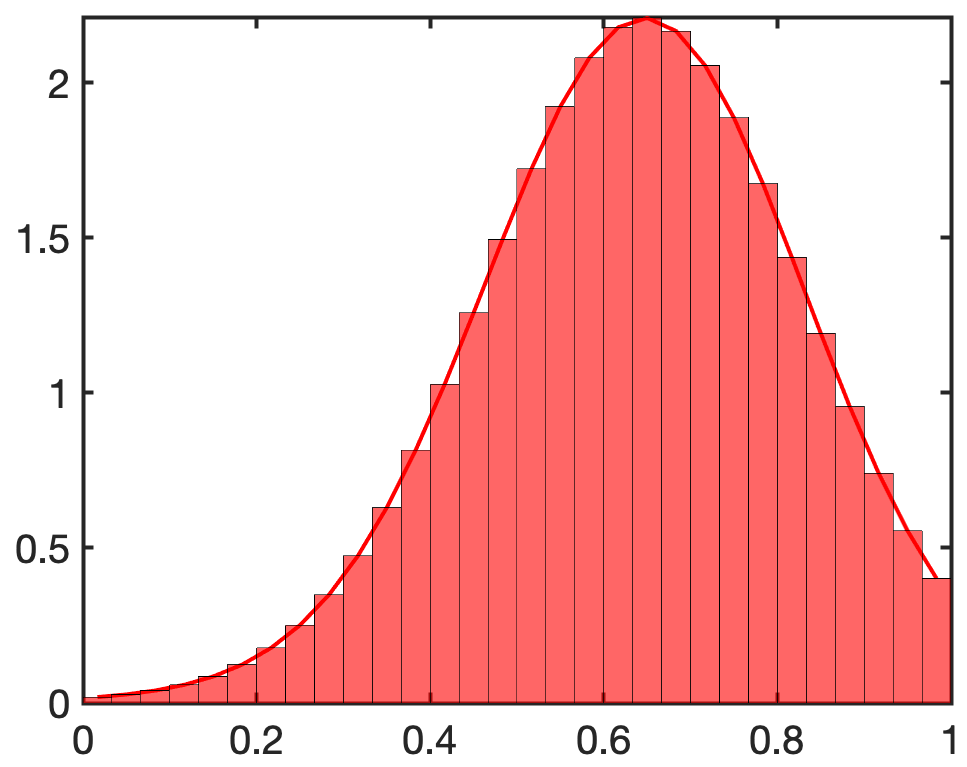}}} \hspace*{10pt}
\subfigure[$N=5,~n=100,~\epsilon=0.1$]{ \resizebox{!}{2.9cm}{\includegraphics{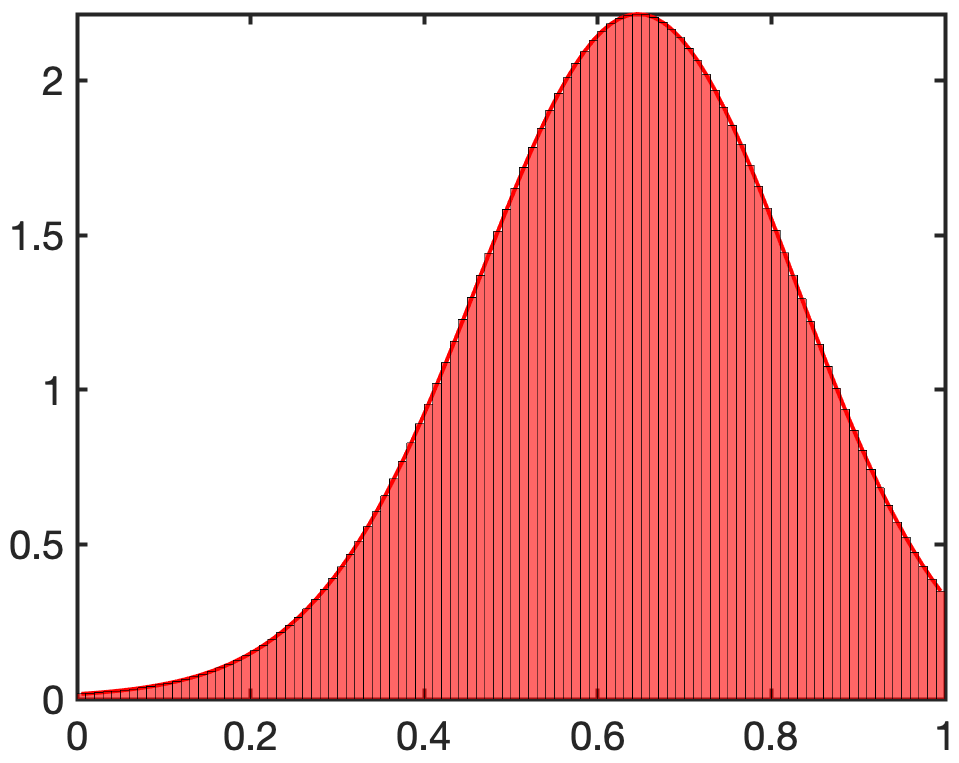}}} \hspace*{10pt}
\subfigure[$N=5,~n=200,~\epsilon=0.1$]{ \resizebox{!}{2.9cm}{\includegraphics{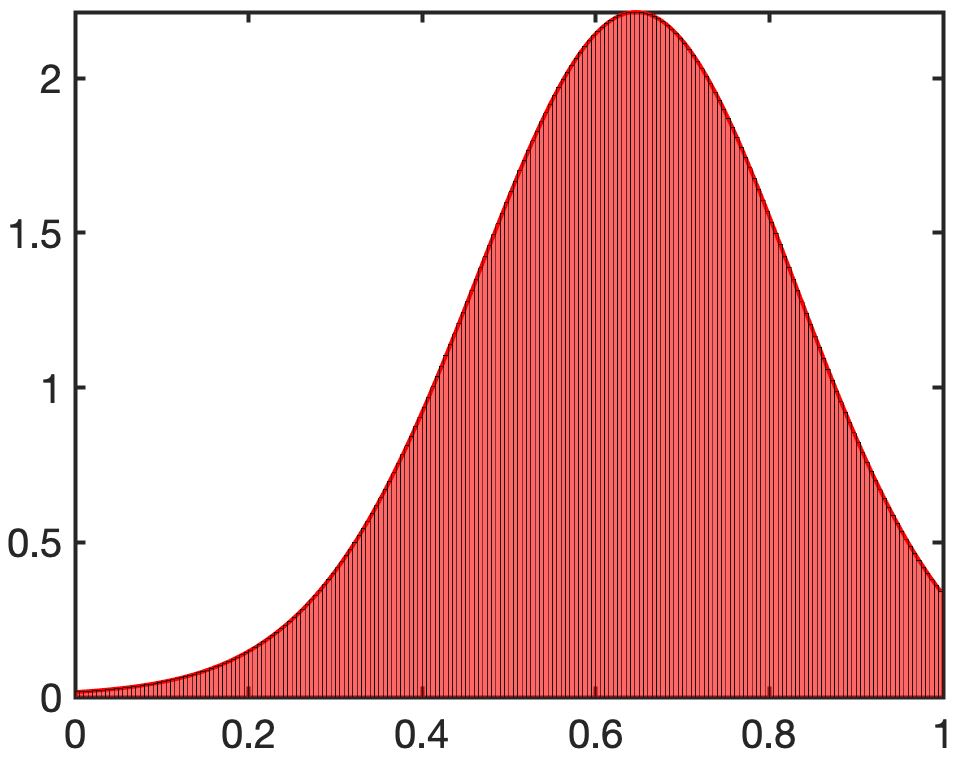}}}
\subfigure[$N=3,~n=100,~\epsilon=0.05$]{ \resizebox{!}{2.9cm}{\includegraphics{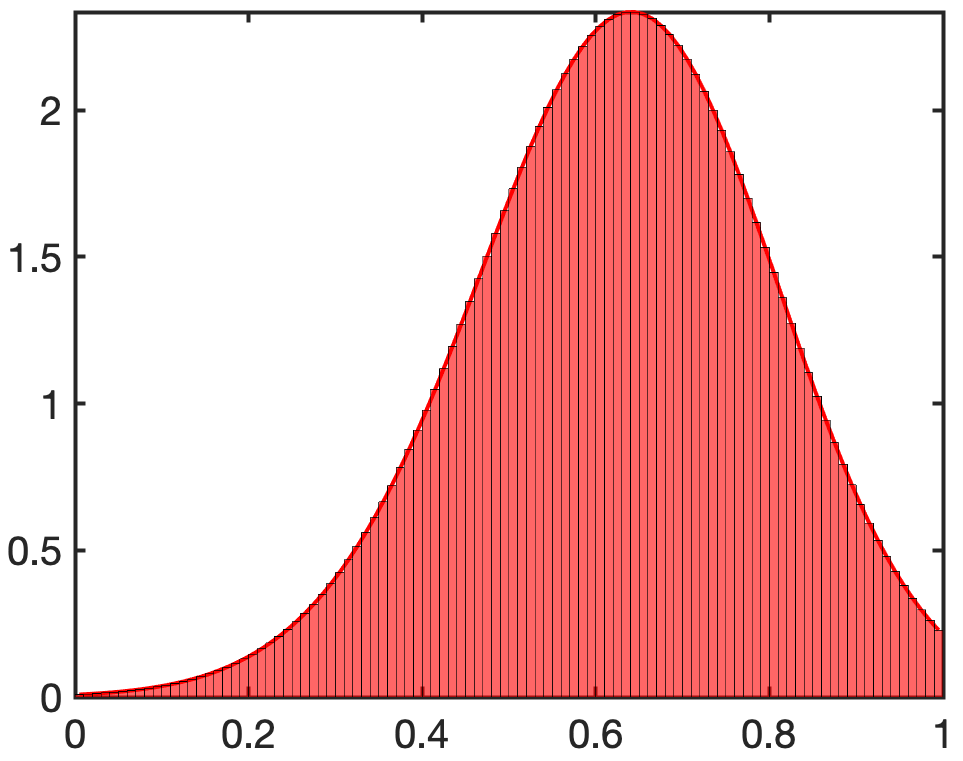}}}\hspace*{10pt}
\subfigure[$N=6,~n=100,~\epsilon=0.05$]{ \resizebox{!}{2.9cm}{\includegraphics{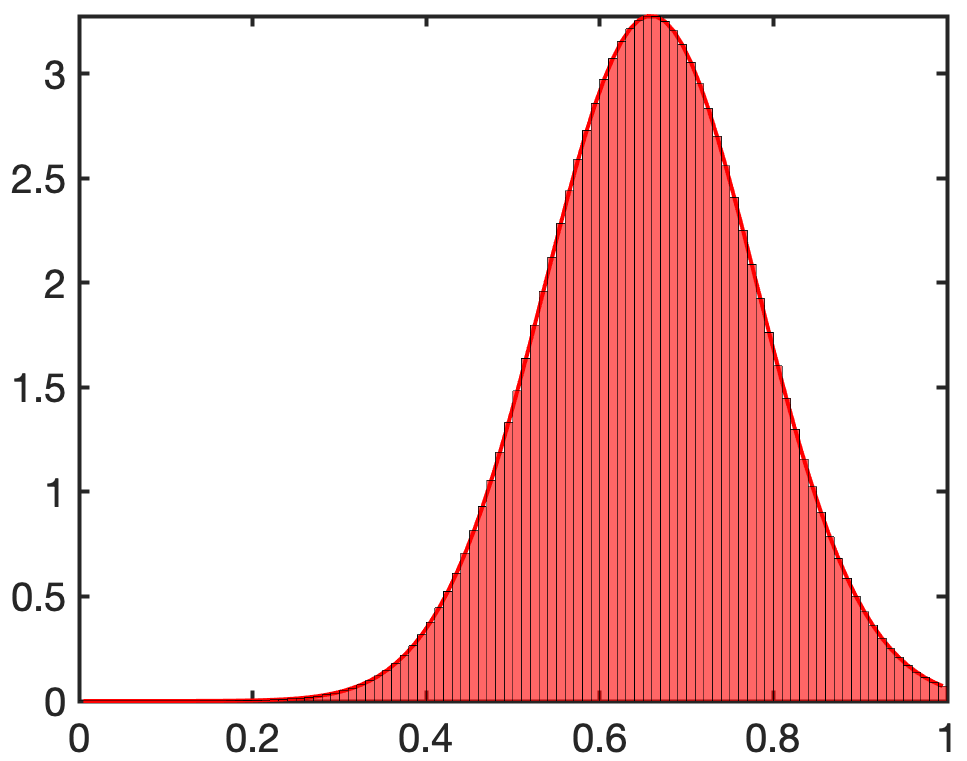}}}\hspace*{10pt}
\subfigure[$N=9,~n=100,~\epsilon=0.05$]{ \resizebox{!}{2.9cm}{\includegraphics{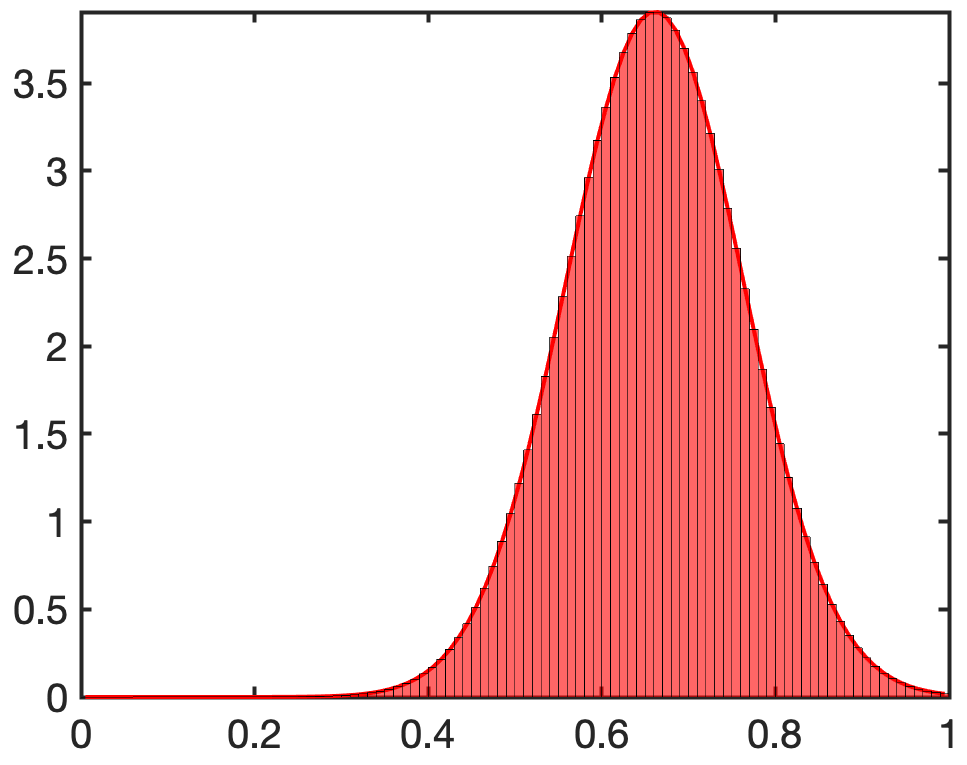}}}
\subfigure[$N=6,~n=100,~\epsilon=0.2$]{ \resizebox{!}{2.9cm}{\includegraphics{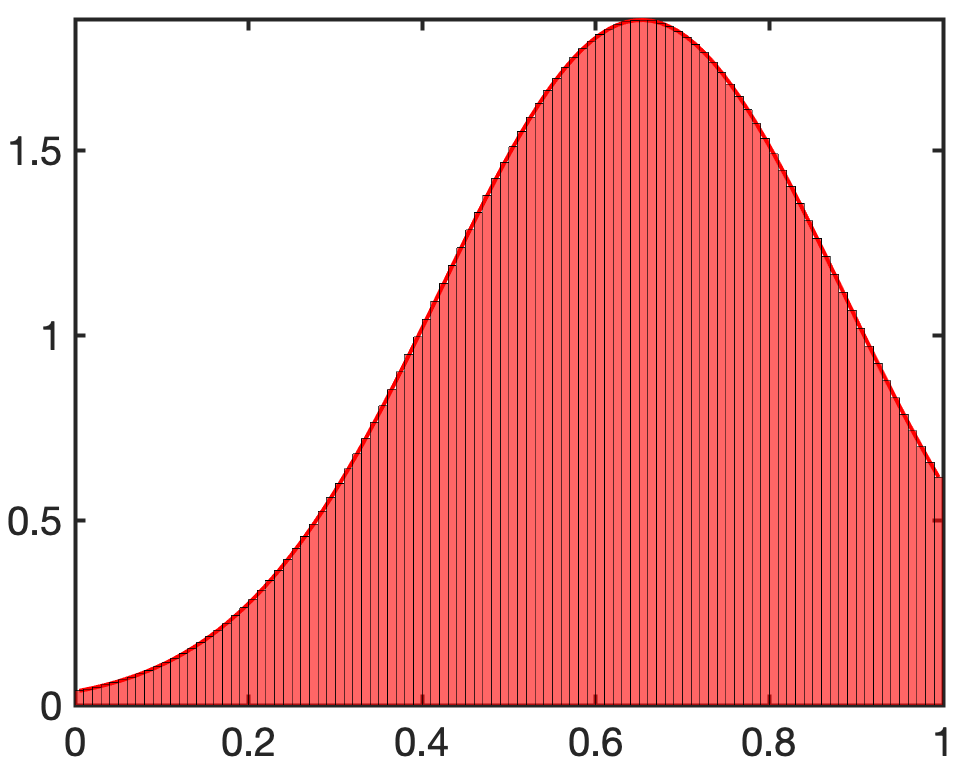}}}\hspace*{10pt}
\subfigure[$N=6,~n=100,~\epsilon=0.1$]{ \resizebox{!}{2.9cm}{\includegraphics{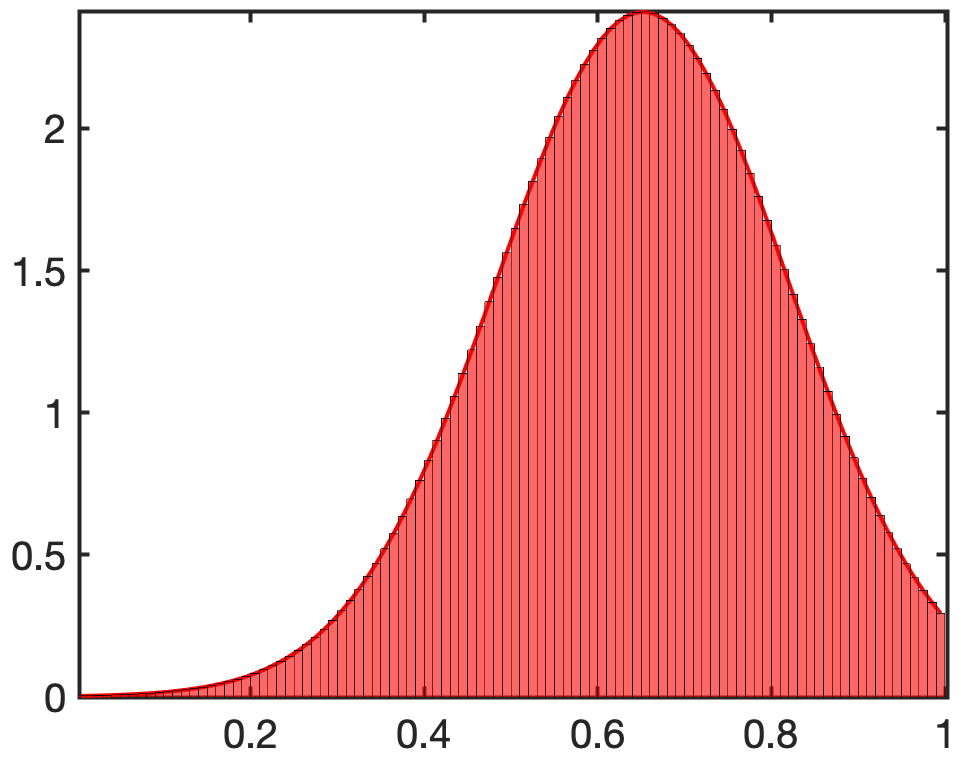}}}\hspace*{10pt}
\subfigure[$N=6,~n=100,~\epsilon=0.03$]{ \resizebox{!}{2.9cm}{\includegraphics{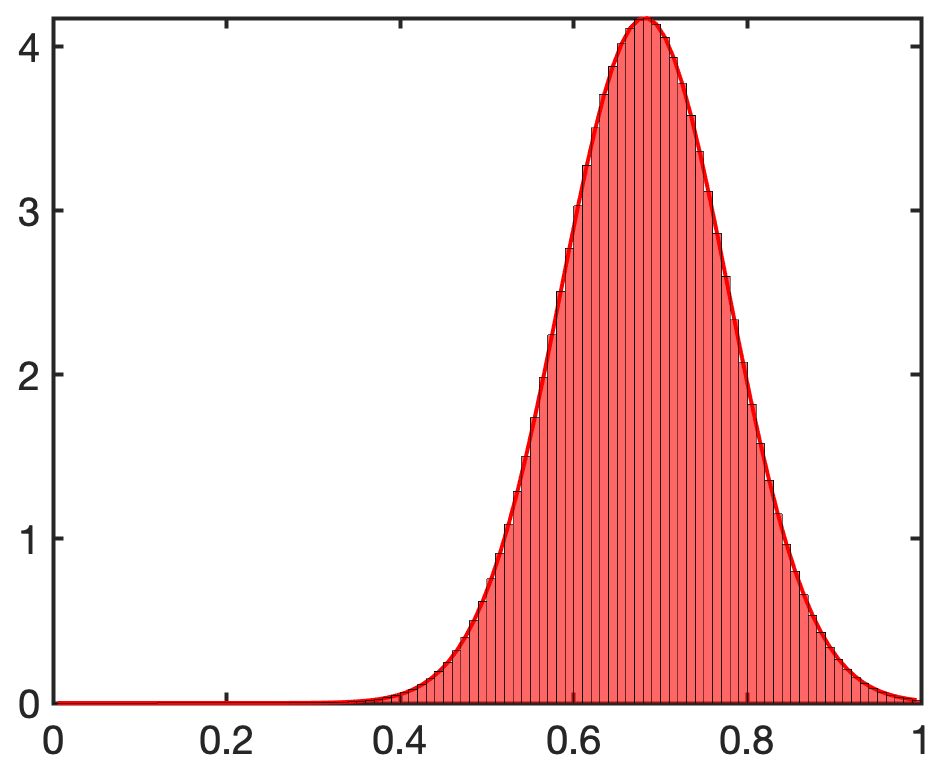}}}
\caption{The stationary distribution of the Markov chain (histogram and red fitting curve) for logistic map $x_{k+1}=3x_k(1-x_k)$. The figures show the concentration of stationary distribution around the single stable fixed point of logistic map at $x=\frac{2}{3}$ for different values of $N$, $n$, and $\epsilon$. }
\label{logistic_map_r3}
\end{figure}

Figure \ref{logistic_r4} depicts the approximated invariant measure for the logistic map where $r=4$. In this case $[0,1]$ is partitioned into 50 equi-length sub-intervals and we construct 5 distributional data by the iterates of logistic map starting from a uniform distribution. The blue curve represents the analytic invariant measure given in \eqref{stationary_logistic}.

\begin{figure}[htb]
	\centering
	\includegraphics[width=1.7in]{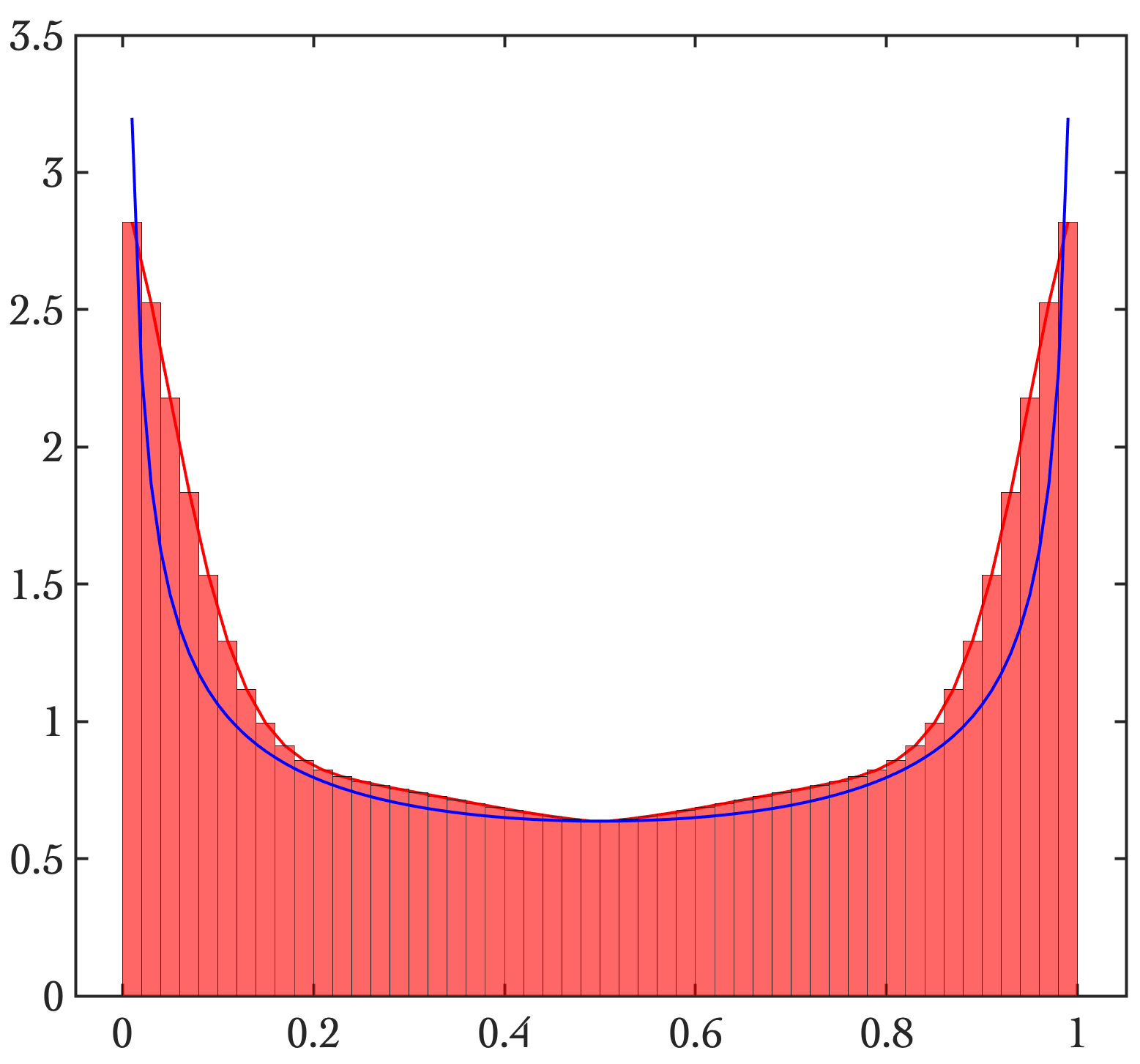}
	\caption{The stationary distribution of the Markov chain (histogram and red fitting curve) is compared to the invariant density of the logistic map for $r=4$ (blue).}
	\label{logistic_r4}
\end{figure}

\section{CONCLUDING REMARKS}\label{sec:conclusion}
In this paper we presented an approach to estimate flow from distributional data. It can be seen as a generalization of Euclidean regression to the Wasserstein space relying on measure-valued curves. It represents a relaxation of geodesic regression in Wasserstein space. The apparently nonlinear primal problem and be recast as a multi-marginal optimal transport, leading to a formulation as a linear program. Entropic regularization and a generalized Sinkhorn algorithm can be effectively employed to solve this multi-marginal problem. 

The proposed framework can be used to estimate correlation between given distributional snapshots. Potential applications of the theory are envisioned to aggregate data inference~\cite{haasler2019estimating}, estimating meta-population dynamics~\cite{nichols2017using}, power spectra tracking~\cite{jiang2011geometric}, and more generally, system identification~\cite{karimi2020data}.
The framework encompasses the case where probability laws are sought for dynamical systems, generating curves to approximate data sets. Future research along these lines, of utilizing higher-order curves and general dynamics, should prove useful in application that may include weather prediction, modeling traffic, besides more traditional ones in computer vision.

\bibliographystyle{siamplain}
\bibliography{references}
\end{document}